\documentclass[a4paper,UKenglish]{mystyle}

\usepackage{latexsym}
\usepackage{amsmath,amssymb}
\usepackage{wrapfig}
\usepackage{atbeginend}
\usepackage{enumitem}
\usepackage{tikz}
\usepackage{hyperref}
\usepackage{bbding}
\usepackage{microtype}
\usepackage{xcolor}
\usepackage{caption}

\usetikzlibrary{arrows,automata,backgrounds,calc,chains,fadings,fit,patterns,positioning,shapes.geometric,shapes.symbols,through,trees,decorations.markings,decorations.text}

\usepackage{url}
\newcommand{\xtend}[1]{}
\newcommand{\xreplace}[2]{#1}

\renewcommand{\theenumi}{(\roman{enumi})}

\renewcommand{\theenumii}{(\alph{enumii})}

\newcommand{\tup}[1]{\langle#1\rangle}
\newcommand{\pair}[2]{\tup{#1,#2}}

\newcommand{\msf}{\mathsf}
\newcommand{\mbb}{\mathbb}
\newcommand{\mbf}{\mathbf}
\newcommand{\mbs}{\boldsymbol}
\newcommand{\mcl}{\mathcal}

\newcommand{\nat}{\mbb{N}}
\newcommand{\rat}{\mbb{Q}}

\newcommand{\reals}{\mbb{R}}

\newcommand{\myal}[2]{\forall #1\;#2}

\renewcommand{\emptyset}{\varnothing}
\newcommand{\emptyword}{\varepsilon}

\newcommand{\length}[1]{|#1|}

\newcommand{\fap}[2]{#1(#2)}
\newcommand{\fapn}[2]{\fap{#1^{#2}}}
\newcommand{\bfap}[3]{\fap{#1}{#2,#3}}

\newcommand{\where}{\mid}

\newcommand{\two}{\mbf{2}}
\newcommand{\str}[1]{#1^{\nat}}

\newcommand{\botdeg}{\mbs{0}}

\newcommand{\seq}[1]{\langle#1\rangle}
\newcommand{\sshift}{\mcl{S}}
\newcommand{\shift}{\fapn{\sshift}}

\newcommand{\wof}[2]{#1 \cdot #2}

\newcommand{\wprod}[2]{#1 \otimes #2}

\newcommand{\zip}{\fap{\msf{zip}}}

\newcommand{\scyc}{\varphi}

\newcommand{\cyc}{\bfap{\scyc}}

\newcommand{\prefixof}{\sqsubseteq}

\newcommand{\morse}{\msf{T}}
\newcommand{\perioddoubling}{\msf{P}}
\newcommand{\mephisto}{\msf{W}}
\newcommand{\uncomp}{\msf{U}}
\newcommand{\comptop}{\msf{C}}

\tikzset{state/.style={draw=black,ellipse,inner sep=.6mm,outer sep=.5mm}}
\tikzset{default/.style={->,>=stealth',shorten >=1pt,shorten <= 1pt,auto,node di
stance=2cm,semithick}}

\tikzstyle{roundNode}=[gyellow,thick,circle,minimum size=4mm,inner sep=0.5mm]

\tikzset{state/.style={draw=black,ellipse,inner sep=.6mm,outer sep=.5mm}}
\tikzset{default/.style={->,>=stealth',shorten >=1pt,shorten <= 1pt,auto,node di
stance=2cm,semithick}}
\tikzset{bl/.style={below left of=#1,yshift=3mm}}
\tikzset{br/.style={below right of=#1,yshift=3mm}}

\tikzset{lbr/.style={below right,inner sep=0.5mm}}
\tikzset{lar/.style={above right,inner sep=0.5mm}}
\tikzset{lbl/.style={below left,inner sep=0.5mm}}
\tikzset{lal/.style={above left,inner sep=0.5mm}}

\tikzset{tloop/.style={out=60,in=120,looseness=5}}
\tikzset{bloop/.style={out=-60,in=-120,looseness=5}}
\tikzset{lloop/.style={out=210,in=150,looseness=5}}
\tikzset{rloop/.style={out=-30,in=30,looseness=5}}

\tikzset{lhead/.style={at=(#1.west),anchor=east,xshift=0cm,inner sep=.5mm}}
\tikzset{rhead/.style={at=(#1.east),anchor=west,xshift=0cm,inner sep=.5mm}}
\tikzset{bhead/.style={at=(#1.south),anchor=north,xshift=0cm,inner sep=.5mm}}

\tikzstyle{gyellow}=[draw=brown!80,top color=yellow!50,bottom color=yellow!20!brown]
\tikzstyle{gblue}=[draw=blue!50,top color=white,bottom color=blue!60]
\tikzstyle{gred}=[draw=red!50,top color=white,bottom color=red!60]
\tikzstyle{ggreen}=[draw=green!70!black,top color=white,bottom color=green!80!black]

\tikzstyle{roundNode}=[gyellow,thick,circle,minimum size=4mm,inner sep=0.5mm]

\newcommand{\sumlength}[1]{||#1||}

\newcommand{\rpl}{{|}}

\newcommand{\wordemp}{\emptyword}

\newcommand{\thuemorse}{Thue\hspace*{.5pt}--Morse}

\newcommand{\alphin}{\Sigma}
\newcommand{\alphout}{\Gamma}
\newcommand{\states}{Q}
\newcommand{\istate}[1]{q_{#1}}
\newcommand{\stransfun}{\delta}
\newcommand{\transfun}[2]{\stransfun(#1,#2)}
\newcommand{\soutfun}{\lambda}
\newcommand{\outfun}[2]{\soutfun(#1,#2)}

\newcommand{\sixtuple}[6]{\tup{#1,#2,#3,#4,#5,#6}}
\newcommand{\quadruple}[4]{\tup{#1,#2,#3,#4}}

\tikzset{degree/.style={circle,
    draw=blue!50,top color=blue!80!black!10!white,bottom color=blue!90!black!60!white,outer sep=.5mm}}

\newcommand{\astr}{w}
\newcommand{\bstr}{u}

\newcommand{\red}[1]{\ge_{\text{#1}}}
\newcommand{\redi}[1]{\le_{\text{#1}}}

\newcommand{\redistrict}[1]{<_{\text{#1}}}

\newcommand{\fstred}{\red{}}
\newcommand{\fstredi}{\redi{}}

\newcommand{\fstredistrict}{\redistrict{}}
\newcommand{\fstconv}{\equiv}
\newcommand{\convclass}[1]{#1^{{\fstconv}}}
\newcommand{\convclassset}[1]{#1/_{{\fstconv}}}

\newcommand{\words}{\str{\two}}
\renewcommand{\deg}{\mathfrak}

\newcommand{\yes}{
  \begin{tikzpicture}[baseline=-.5ex,scale=.9,nodes={scale=.9}]
    \node [circle,draw=none,fill=green!70!black!30,inner sep=0mm,minimum size=5mm] {{\small \CheckmarkBold}};
  \end{tikzpicture}%
}
\newcommand{\no}{
  \begin{tikzpicture}[baseline=-.5ex,scale=.9,nodes={scale=.9}]
    \node [circle,draw=none,fill=red!90!black!30,inner sep=0mm,minimum size=5mm] {{\small \XSolidBrush}};
  \end{tikzpicture}%
}
\newcommand{\open}{
  \begin{tikzpicture}[baseline=-.5ex,scale=.9,nodes={scale=.9}]
    \node [circle,draw=none,fill=blue!80!green!30,inner sep=0mm,minimum size=5mm] {\!\textbf{?}};
  \end{tikzpicture}%
}

\newcommand{\questionnodee}[1]{%
  \begin{tikzpicture}[baseline=-0.5ex]
  \node [circle,ultra thick,draw=blue!70!green!70,fill=blue!70!green!10,inner sep=.2mm,outer sep=0mm] {{\small #1}};
  \end{tikzpicture}%
}

\newcommand{\questionnode}[1]{%
  \questionnodee{Q\arabic{#1}}%
}

\makeatletter
\def\jlabel#1#2{{\def\@currentlabel{#2}\label{#1}}}
\makeatother

\newcounter{q}
\newcommand{\myquestion}{\addtocounter{q}{1}\scalebox{.8}{\questionnode{q}}}
\newcommand{\qlabel}[1]{
  \newcounter{#1}
  \setcounter{#1}{\value{q}}
}

\newcommand{\concat}{\mathrel{;}}

\newcommand{\mmul}{\,}

\usepackage{euler}

\renewcommand{\xtend}[1]{#1}
\renewcommand{\xreplace}[2]{#2}

\title{Degrees of Infinite Words, \protect\\ Polynomials, and Atoms \protect\\
  (Extended Version)\footnote{%
  This research has been supported by the Academy of Finland  under the grant 257857.%
}\footnote{%
  This is an extended version of the papers~\cite{endr:karh:klop:saar:2016,endr:karh:klop:saar:2018}.%
}}

\titlerunning{Degrees of Infinite Words, Polynomials and Atoms (Extended Version)}

\author[1]{J\"{o}rg Endrullis}
\author[2]{Juhani Karhum\"{a}ki}
\author[1,3]{\\Jan Willem Klop}
\author[2]{Aleksi Saarela}

\affil[1]{
  Department of Computer Science \\
  VU University Amsterdam,
  Amsterdam, the Netherlands \\
  Email: j.endrullis@vu.nl,\; j.w.klop@vu.nl
}
\affil[2]{
  Department of Mathematics and Statistics \& FUNDIM\\
  University of Turku,
  Turku, Finland \\ 
  Email: karhumak@utu.fi,\; amsaar@utu.fi
}
\affil[3]{
  Centrum Wiskunde \& Informatica (CWI),
  Amsterdam, the Netherlands
}

\authorrunning{J. Endrullis, J. Karhum\"{a}ki, J.W. Klop and A. Saarela}

\keywords{finite state transducers, infinite words, infinite sequences, streams, complexity, degrees}

\serieslogo{}
\EventShortName{}

\def\copyrightline{}
                             
\begin{document}
\captionsetup{singlelinecheck=true,margin=1cm}
\maketitle

\begin{abstract}
We study finite-state transducers and their power 
for transforming infinite words. 
Infinite sequences of symbols are of paramount importance in a wide range of fields, 
from formal languages to pure mathematics and physics. 
%
While finite automata for recognising and transforming languages are well-understood, 
very little is known about the power of automata to transform infinite words. 

The word transformation realised by finite-state transducers
gives rise to a complexity comparison of words
and thereby induces equivalence classes, called \emph{(transducer) degrees}, 
and a partial order on these degrees.
The ensuing hierarchy of degrees is analogous to  
the recursion-theoretic \emph{degrees of unsolvability}, also known as \emph{Turing degrees},
where the transformational devices are Turing machines.
However, as a complexity measure, Turing machines are too strong: 
they trivialise the classification problem by identifying all computable words. 
Finite-state transducers give rise to a much more fine-grained, discriminating hierarchy. 
In contrast to Turing degrees, 
hardly anything is known about transducer degrees,
in spite of their naturality.

We use methods from linear algebra and analysis 
to show that there are infinitely many atoms in the transducer degrees,
that is, minimal non-trivial degrees.
\xtend{%
We also show that there exists an uncomputable degree
that has only uncomputable degrees and the trivial bottom degree below itself.
}

%
\end{abstract}
\allowdisplaybreaks

\section{Introduction}\label{sec:intro}

In recent times, computer science, logic and mathematics have extended
the focus of interest from finite data types to include infinite data types,
of which the paradigm notion is that of infinite sequences of symbols, or \emph{words}.
Infinite words are of paramount importance in a wide range of fields,
from formal languages to pure mathematics and physics:
they appear in functional programming, formal language
theory, in the mathematics of dynamical systems, fractals and number theory, in
business (financial data words) and in physics (signal processing).

\newcommand{\tree}{
    \begin{scope}[line width=2mm,blue!80!green!40,xshift=50mm,yshift=-10mm,xscale=.8]
      \draw (0,0) to[out=90,in=-90] 
        node [at end, above] {mathematics} 
        node [pos=.92,s] {number theory} 
        node [pos=.79,s] {dynamic systems} 
        node [pos=.65,s] {fractals} 
        (-4cm,6cm);
      \draw (0,0) to[out=90,in=-110] 
        node [at end, above] {functional\\programming} 
        node [pos=.85,s] {communicating\\processes} 
        (1cm,6cm);
      \draw (0,0) to[out=90,in=-90] 
        node [at end, above] {physics,\\ engineering} 
        node [pos=.9,s] {embedded\\devices} 
        node [pos=.7,s] {signal\\processing} 
        node [pos=.5,s] {financial\\data words} 
        (4cm,6cm);
      \draw (0,0) to[out=90,in=-110] 
        node [at end, above] {formal\\languages} 
        node [pos=.45,xshift=-.5cm,s] {combinatorics\\ on words} 
        (-1.5cm,6cm);
    \end{scope}
}

\definecolor{grasstop}{HTML}{648510}
\definecolor{grassbottom}{HTML}{8AB51F}

\definecolor{mountainlight}{HTML}{857a2b}
\definecolor{mountaindark}{HTML}{534307}

\definecolor{snowlight}{HTML}{ece9c0}
\definecolor{snowdark}{HTML}{c8ae59}

\newcommand{\landscape}{
    \begin{scope}
    \clip (0,-2cm) -- (10cm,-2cm) -- (10cm,5cm) -- (0,5cm);
          
    \begin{scope}[shift={(2cm,3cm)},m/.style={draw=none,top color=mountainlight!30,bottom color=mountaindark!30}]
    \draw [rounded corners=2mm,m] 
      (0,0) to (1.2cm,-.8cm) to (1.5cm,-1.5cm) to (3cm,-4cm) to (-3cm,-4cm) to (-1.5cm,-1.5cm) -- cycle;
    \end{scope}
    \begin{scope}[shift={(2cm,3cm)},m/.style={draw=none,top color=snowlight,bottom color=white}]
    \clip [rounded corners=2mm] 
      (0,0) to (1.2cm,-.8cm) to (1.5cm,-1.5cm) to (3cm,-4cm) to (-3cm,-4cm) to (-1.5cm,-1.5cm) -- cycle;
    \node [starburst,minimum size=30mm,m] {}; 
    \end{scope}

    \begin{scope}[shift={(8.5cm,2.5cm)},xscale=.7,m/.style={draw=none,top color=mountainlight!20,bottom color=mountaindark!20}]
    \draw [rounded corners=2mm,m] 
      (0,0) to (1.2cm,-.8cm) to (1.5cm,-1.5cm) to (3cm,-4cm) to (-3cm,-4cm) to (-1.5cm,-1.5cm) -- cycle;
    \end{scope}
    \begin{scope}[shift={(8.5cm,2.5cm)},xscale=.7,m/.style={draw=none,top color=snowlight,bottom color=white}]
    \clip [rounded corners=2mm] 
      (0,0) to (1.2cm,-.8cm) to (1.5cm,-1.5cm) to (3cm,-4cm) to (-3cm,-4cm) to (-1.5cm,-1.5cm) -- cycle;
    \node [starburst,minimum size=30mm,m] {}; 
    \end{scope}

    \begin{scope}[shift={(5.5cm,2.2cm)},xscale=-1,yscale=.7,m/.style={draw=none,top color=mountainlight!50,bottom color=mountaindark!50}]
    \draw [rounded corners=2mm,m] 
      (0,0) to (1.2cm,-.8cm) to (1.5cm,-1.5cm) to (3cm,-4cm) to (-3cm,-4cm) to (-1.5cm,-1.5cm) -- cycle;
    \end{scope}
    \begin{scope}[shift={(5.5cm,2.2cm)},xscale=-1,,yscale=.7,m/.style={draw=none,top color=snowlight,bottom color=white}]
    \clip [rounded corners=2mm] 
      (0,0) to (1.2cm,-.8cm) to (1.5cm,-1.5cm) to (3cm,-4cm) to (-3cm,-4cm) to (-1.5cm,-1.5cm) -- cycle;
    \node [starburst,minimum size=25mm,m] {}; 
    \end{scope}

    \end{scope}

    \draw [draw=none,top color=grasstop!40,bottom color=grassbottom!40,yshift=2mm] 
      (0cm,-.1cm) to[out=20,in=220] (4.5cm,0.3cm) to[out=40,in=180] (7cm,-0cm) to[out=0,in=180] (9cm,.0cm)
      to[out=0,in=180] (10cm,-.1cm)
      to (10cm,-2cm) to (0cm,-2cm) -- cycle;
    \draw [draw=none,top color=grasstop!70,bottom color=grassbottom!70] 
      (0cm,-1cm) to[out=20,in=180] (2.5cm,.2cm) to[out=0,in=180] (5cm,-.8cm) to[out=0,in=180] (9cm,.-.2cm)
      to[out=0,in=180] (10cm,-.2cm)
      to (10cm,-2cm) to (0cm,-2cm) -- cycle;
    \draw [draw=none,top color=grasstop,bottom color=grassbottom] 
      (0cm,0cm) to[out=20,in=180] (3cm,-1cm) to[out=0,in=180] (6cm,-.3cm) to[out=0,in=180] (10cm,-1cm)
      to (10cm,-2cm) to (0cm,-2cm) -- cycle;
      
    \begin{scope}[yshift=-5mm]
    \tree
    \end{scope}
}

\xtend{%
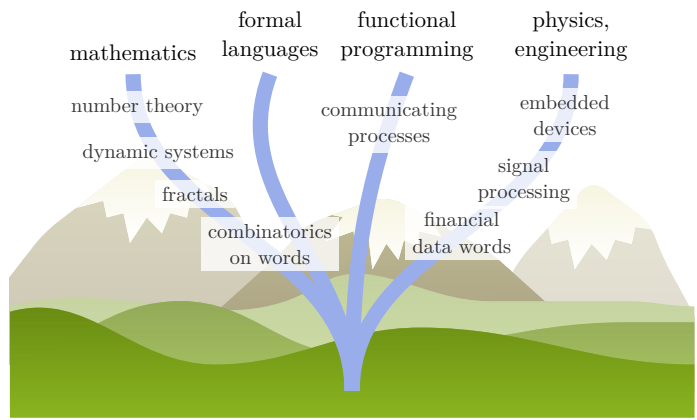
\begin{figure}[h!]
  \centering
  \begin{tikzpicture}[xscale=0.9,yscale=.7,
    very thick,nodes={black,fill=white,align=center},nodes={scale=.85},s/.style={scale=.9,opacity=.8}]
    
    \landscape
  \end{tikzpicture}  
  \caption{Panta Rhei; words are ubiquitous.}
  \label{fig:panta}
\end{figure}
}

An accepted and deep mathematical insight is that together with a class of structures, one has to deal
with the ways to transform these structures into each other, such as morphisms%
\footnote{Morphisms in a category should not be confused with morphisms in the sense of this paper.} in a category.
Our objects of interest are words and how they can be transformed into each other
via finite-state transducers.

\xtend{%
In computer science, infinite words are often referred to as streams. 
In the real world, we encounter 
streams in the form of sensor data from continual measurements, 
streams of financial transactions, and
streams of messages in social media.
Analysing and processing the large amount of data generated by these applications
is one of the major challenges of computer science today,
and an active field of research, known as \emph{Big Data}.
When it comes to data sets that are massive in size, 
even linear algorithms may be too complex for processing the data.
For instance, think of an algorithm with linear space complexity (e.g. linear random-access memory)
applied to petabytes of input data.
This has led to the research field of \emph{sublinear algorithms},
a rapidly developing area of computer science, that is rooted in the study of Big Data.
This field is concerned with the development of algorithms having sublinear-space and/or sublinear-time complexity. 

We are interested in the most strict form of sublinear-space complexity:
\emph{constant-space complexity}.
Algorithms with constant space-complexity ($O(1)$ space-complexity) are indispensable for programs
that are intended to run indefinitely, to continually transform 
an endless input word into an endless output word.
Any algorithm \emph{not} having constant space-complexity will eventually 
run out of memory on a real world device (computer)
when transforming an \emph{infinite} word.
This motivates the study of constant-space algorithms for the transformation of words.
Note that a constant-space algorithm 
is nothing else than a finite-state automaton.
}

While finite automata for recognising and transforming languages
are well-studied and well-understood,
surprisingly, very little is known about the power of finite automata for transforming words.
%
Even for concrete examples of words $w_1$ and $w_2$, 
there exist no techniques to determine whether $w_1$ can be transformed into $w_2$ by some finite-state transducer.
See, e.g., Questions~\questionnodee{Q4} and~\questionnodee{Q5}.
We are interested in understanding the power of finite-state transducers for transforming words.

Such a study can be profitably cast in the form of setting up a hierarchy of degrees, 
induced by a transformational device (sometimes called a `reduction').
This is a well-known reasoning framework in logic and computer science~\cite{loewe}, 
with many instances, e.g., Wadge degrees~\cite{van1978wadge}, 
Turing degrees~\cite{shoe:1971,odif:1999}, r.e. degrees~\cite{odif:1999}, and so on. 
In our case the hierarchy of degrees is obtained as follows.

The transformation realised by finite-state transducers
induces a partial order of degrees of infinite words:
for words $v,w \in \str{\Delta}$, 
we write $v \ge w$ if $v$ can be transformed into $w$ by some finite-state transducer.
If $v \ge w$, then $v$ can be thought of as \emph{at least as complex as} $w$. 
This complexity comparison induces equivalence classes of words, called \emph{degrees}, 
and a partial order on these degrees, that we call \emph{transducer degrees}.
%

The ensuing hierarchy of degrees is analogous to  
the recursion-theoretic \emph{degrees of unsolvability}, also known as \emph{Turing degrees},
where the transformational devices are Turing machines.
The Turing degrees have been widely studied in the 60's and 70's.
However, as a complexity measure, Turing machines are too strong: 
they trivialise the classification problem by identifying all computable infinite words. 
Finite-state transducers (FSTs) give rise to a much more fine-grained hierarchy. 

In our view, transducers are the most natural devices for transforming words.
Unlike Turing machines, they are not too strong and still very expressive.
On the one hand, transducers are `weak enough' to exhibit a rich structure within the computable words. 
On the other hand, they capture several usual transformations, such as alphabet renaming, 
insertion and removal of elements, or morphisms as usually studied in theories of infinite sequences~\cite{allo:shal:2003}.
%

\begin{wrapfigure}{r}{.35\textwidth}
  \vspace{-2ex}
  \begin{tikzpicture}[thick,>=stealth,
    node/.style={degree}]
  
    \node (evp) [node,minimum size=6mm] at (0mm,0mm) {$\botdeg$};
    \node [anchor=west] at (3mm,0mm) {\textit{bottom degree}};
  
    \node at ($(evp)+(0mm,9mm)$) {other degrees};
  
    \begin{scope}[very thick,dashed]
    \draw (evp) -- ($(evp) + (15mm,10mm)$);
    \draw (evp) -- ($(evp) + (-15mm,10mm)$);
    \end{scope}
  \end{tikzpicture}\vspace{-1ex}
\end{wrapfigure}
Like the Turing degrees, the transducer degrees have a bottom degree
that is less than or equal to all other degrees (pictorial on the right).
The bottom degree of the Turing degrees contains all computable words.
In contrast, transducer degrees are much more fine-grained.
The \emph{bottom degree} $\botdeg$ of the transducer degrees consists only of the ultimately periodic words,
that is, words of the form $uvvv\cdots$ for finite words $u,v$.

We present a comparison of some basic properties,
as to their validity in the Turing degrees and the transducer degrees. 
\xreplace{%
  For Turing degrees we have~\cite{shoe:1971,odif:1999}
}{%
  We list a few key results~\cite{shoe:1971,odif:1999} on Turing degrees 
  due to Spector, Kleene, Post, Sacks, Lacombe and Simpson:
}
  \begin{enumerate}\setlength{\itemsep}{.2ex}
    
    \item\label{thm:tm:atom} 
      There exist $2^{\aleph_0}$ atom (minimal) degrees.
      \hfill\open
    \item\label{thm:tm:cover} 
      Every degree has a minimal cover.
      \hfill\open
    \item\label{thm:tm:fin:sup} 
      Every finite set of degrees has a supremum.
      \hfill\no
    \item\label{thm:tm:inf:sup} 
      No infinite ascending sequence has a supremum.
      \hfill\open
    \item\label{thm:tm:fin:inf} 
      There are pairs of degrees without infimum.
      \hfill\yes
    \item\label{thm:tm:incompare} 
      For every degree $\ne \botdeg$ there exists an incomparable degree.
      \hfill\yes
    \item\label{thm:tm:embed} 
      Every countable partial order can be embedded.
      \hfill\open
    \item\label{thm:tm:dense} 
      The recursively enumerable degrees are dense.
      \hfill\no
    \item\label{thm:tm:logic} 
      The first-order theory of Turing degrees in the language 
      $\pair{\red{}}{=}$ is recursively isomorphic to that of true second-order arithmetic.
      \hfill\open
  \end{enumerate}
Here the symbols on the right indicate whether the properties also hold for transducer degrees:
\yes{} if the property also holds for transducer degrees,
\no{} if it fails,
and \open{} for questions that are open in transducer degrees.

In previous papers~\cite{endr:hend:klop:2011,endr:hend:2014,endr:grab:hend:zant:2015,endr:klop:saar:whit:2015}, 
we have discussed several structural properties of the hierarchy of transducer degrees.
In this paper, we focus on atom degrees.
An \emph{atom degree} is a minimal non-trivial degree,
that is, a degree that is directly above the bottom degree without interpolant degree:
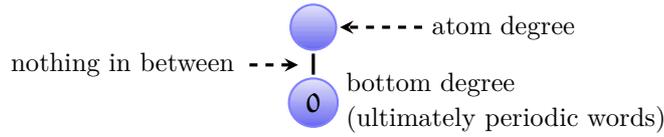
\begin{center}\vspace{-1ex}
  \begin{tikzpicture}[thick,>=stealth,
    node/.style={degree}]
  
    \node (evp) [node,minimum size=6mm] at (0mm,0mm) {$\botdeg$};
    \node [anchor=west,align=left] at (3mm,0mm) {bottom degree\\(ultimately periodic words)};
  
    \node (prime) [node,minimum size=6mm] at (0mm,10mm) {};
    \node (primeL) [anchor=west,node distance=25mm,right of=prime] {\textit{\emph{atom degree}}};
    \draw [->,very thick,dashed] (primeL) -- (prime);

    \draw [<-,very thick,dashed] ($(evp)!.5!(prime)+(-2mm,0mm)$) -- node [anchor=east,at end] {nothing in between} +(-7mm,0mm);
  
    \begin{scope}[very thick]
    \draw (prime) -- (evp);
    \end{scope}
  \end{tikzpicture}
\end{center}
Thus the atom degrees reduce only to $\botdeg$ or themselves.
The following questions are still open for transducer degrees:
  \;\myquestion\qlabel{q:atoms:uncountable}\!Are there $2^{\aleph_0}$ atoms in the transducer degrees?
  \;\myquestion\qlabel{q:atoms:uncomputable}\!Do uncomputable atoms exist in the transducer degrees?
  \;\myquestion\qlabel{q:atoms:morse}\!Is the degree of the \thuemorse{} word $\morse = 0110100110010110 \cdots$ an atom?

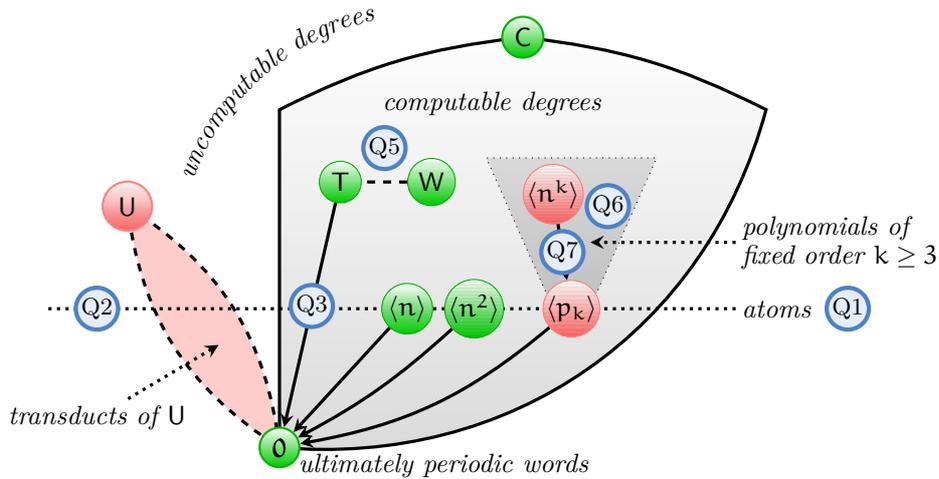
\begin{figure*}[h!]
  \centering
  \begin{tikzpicture}[thick,>=stealth,
  node/.style={circle,fill=black,draw=none,ggreen},scale=.8,inner sep=.75mm,
  r/.style={gred}]

    \begin{scope}[scale=.8]
    \draw [very thick,black,shade,top color=black!1,bottom color=black!15] (0mm,0mm) to[bend left=0] (0mm,70mm) to[bend left=10] (50mm,85mm) to[bend left=10] (100mm,70mm) to[bend left=40] (0mm,0mm);
    \node [node,minimum size=5mm] at (50mm,85mm) {$\comptop$};
    \end{scope}
    \node [anchor=center] at (35mm,57mm) {\textit{computable degrees}};
    \path [postaction={decorate,decoration={text along path,text align=center,text={|\itshape|uncomputable degrees}}}] (-15mm,30mm) to[bend left=60] (30mm,70mm);
  
  
  
  
    \node (u) [node,r,minimum size=6.5mm,inner sep=0mm] at (-25mm,40mm) {$\uncomp$};
    \node (uL) [anchor=center,align=center] at (-30mm,5mm) {\textit{transducts of $\uncomp$}};
    \draw [very thick,dotted,->] (uL) to ($(uL) + (20mm,12mm)$);

    \node (evp) [node,minimum size=5mm] at (0mm,0mm) {$\botdeg$};
    \node [anchor=west] at (2mm,-3mm) {\textit{ultimately periodic words}};
    
    \begin{pgfonlayer}{background} 
    \draw [very thick,black,dashed,fill=red!20] (-25mm,40mm) to[bend left=25] ($(evp)$) to[bend left=25] cycle;
    \end{pgfonlayer}
    
  \begin{scope}[xshift=45mm,yshift=10mm]
    \node (morse) [node,minimum size=5mm] at (-35mm,34mm) {$\morse$};
    \node (sierpinski) [node,minimum size=5mm] at (-20mm,34mm) {$\mephisto$};
  
      

    \node (primeL) [anchor=west] at (30mm,13mm) {\textit{atoms}};
    \draw [very thick,dotted] (primeL) -- ($(primeL)+(-120mm,0mm)$);
    \node at (primeL.east) [anchor=west] {\questionnodee{Q1}} ;

    \begin{scope}[xshift=-15mm]
    \draw [draw=black,dotted,thin,shade,top color=black!15,bottom color=black!25] 
      (15mm,13mm) to ++(-11mm,25mm) to ++(28mm,0mm) to (21mm,13mm) -- cycle;

    \node (primen) [node,minimum size=6.5mm,inner sep=.2mm] at (-9mm,13mm) {$\seq{n}$};
    \node (primen2) [node,minimum size=6.5mm,inner sep=0mm] at (2mm,13mm) {$\seq{n^2}$};
    \node (primenk) [node,r,minimum size=6.5mm,inner sep=0mm] at (18mm,13mm) {$\seq{p_k}$};
    \node (nk) [node,r,minimum size=6.5mm,inner sep=0mm] at (15mm,32mm) {$\seq{n^k}$};
    \node at (24mm,30mm) {\questionnodee{Q6}};
    \end{scope}
    
    \node (poly) [anchor=west,align=left] at (30mm,24mm) {\textit{polynomials of}\\[-.5ex]\textit{fixed order $k \ge 3$}};
    \draw [->,very thick,dotted] (poly) -- ($(poly)+(-41mm,0mm)$);
  \end{scope}
  
    \node at (-30mm,23mm) {\questionnodee{Q2}} ;

    \begin{scope}[very thick,->]
    \draw (primen) to[bend left=0] (evp);
    \draw (primen2) to[bend left=7] (evp);
    \draw (primenk) to[bend left=15] (evp);
    \draw (nk) -- node {\questionnodee{Q7}} (primenk);
    \draw (morse) -- node [pos=.46] (mp) {\questionnodee{Q3}} (evp); 
    \draw [-,dashed] (sierpinski) -- node [above,solid] {\questionnodee{Q5}} (morse);
    \end{scope}
\end{tikzpicture}
  \caption{\textit{The partial order of transducer degrees with focus on the properties studied in this paper. 
    Our contribution is indicated using the colour red.
    Here $p_k$ is a particular polynomial of order $k$, 
    see Section~\ref{sec:polynomial}. 
    The degree of $\seq{p_k}$ is an atom and
    all other polynomials of order $k$ can be transduced to $\seq{p_k}$.
    For $k \ge 3$, the degree $\seq{n^k}$ is not an atom as shown in Section~\ref{sec:nk}.
    The definitions of the words \thuemorse{}~$\morse$ and the Mephisto Walz~$\mephisto$ are given in Section~\ref{sec:prelims}.
    The degree of $\comptop$ is the top degree of the computable words.
    Finally, the nodes Q1,\ldots,Q7 indicate open problems discussed in Sections~\ref{sec:intro} and~\ref{sec:degrees}.
    }}
  \label{fig:hierarchy}
\end{figure*}%

We show that there are at least $\aleph_0$ atoms
residing in the interesting subclass of words that we call \emph{sporadic words},
of which the simplest one is
$1\,10\,100\,1000\,10000\,\cdots$.
(Jacobs~\cite{jaco:1992} called this word `rarefied ones'.)
Here `sporadic' refers to the fact that the ones are becoming more and more sporadic.
In general, they are of the form
$\seq{f} = 10^{f(0)} \, 1 0^{f(1)} \, 1 0^{f(2)} \, \cdots$,
for some $f : \nat \to \nat$. 
%
This paper studies in particular the case where $f$ is a polynomial.
We consider the `atomicity' of these words 
depending on the polynomials determining how the ones become ever more sporadic.

\xtend{
In this `polynomially sporadic' subhierarchy
of the transducer degrees, we have the following state of affairs:
\begin{itemize}
  \item $\seq{n}$ is an atom, 
  \item $\seq{n^2}$ is an atom,
  \item $\seq{n^3}$ is \emph{not} an atom, 
  \item $\seq{a n^3 + b n^2 + c n + d}$ is an atom for some $a,b,c,d > 0$.
\end{itemize}
This hints at an interesting, rich structure in this hierarchy of sporadic degrees.
}
\vspace{-1ex}

\paragraph*{Our contribution.}

The words $\seq{n}$ and~$\seq{n^2}$ are atoms~\cite{endr:hend:klop:2011,endr:grab:hend:zant:2015}.
Surprisingly, we find that this does not hold for $\seq{n^3}$.
In particular, we show that the degree of $\seq{n^k}$ is \emph{never} an atom for $k \ge 3$
(see Theorem~\ref{thm:not:atom}).
On the other hand, we prove that for every $k > 0$,
there exists a unique atom among the degrees 
of words $\seq{p(n)}$ for polynomials $p(n)$ of order~$k$
(see Theorem~\ref{thm:atom}).
(To avoid confusion between two meanings of \emph{degrees}, namely \emph{degrees of words} and \emph{degrees of polynomials},
we speak of the \emph{order} of a polynomial.)
We moreover show that this atom is the infimum of all degrees of polynomials $p(n)$ of order~$k$.
Figure~\ref{fig:hierarchy} summarises the state of affairs as in this paper.
Finally, we show that there exists an uncomputable word $\uncomp$
that transduces only to uncomputable words or to ultimately periodic words.
\xtend{%
In particular, this word has no transducts that are computable and not ultimately periodic.
}

\paragraph*{Further related work.}

L\"{o}we~\cite{loewe} discussed complexity hierarchies derived from notions of reduction.
The paper~\cite{endr:klop:saar:whit:2015} gives an overview over the subject
of transducer degrees and compares them with the well-known Turing degrees~\cite{shoe:1971,odif:1999}.
Restricting the transducers to output precisely one letter in each step,
we arrive at Mealy machines.
These give rise to an analogous hierarchy of Mealy degrees
that has been studied in~\cite{belo:2008,rayn:1974}.
The structural properties of  
this hierarchy are very different from the transducer degrees~\cite{endr:klop:saar:whit:2015}.
The paper~\cite{bosm:zant:2016} studies a hierarchy of two-sided infinite sequences
arising from the transformation realised by permutation transducers.

\xreplace{
The current paper is an extension of~\cite{endr:karh:klop:saar:2016}.
Beyond better exposition, more examples and more detailed proofs, 
the results in Section~\ref{sec:uncomputable} are new.
}{%
}

\section{Preliminaries}\label{sec:prelims}

Let $\Sigma$ be an alphabet. The empty word is denote by $\emptyword$.
Let $\Sigma^*$ be the set of finite words over~$\Sigma$,
and $\Sigma^+ = \Sigma^* \setminus \{\emptyword\}$.
The set of infinite words over~$\Sigma$ is 
$\str{\Sigma} = \{ \sigma \where \sigma : \nat \to \Sigma \}$
and we let $\Sigma^\infty = \Sigma^* \cup \str{\Sigma}$.
Let $u,w \in \Sigma^\infty$.
Then $u$ is called a \emph{prefix of} $w$, denoted $u \prefixof w$,
if $u = w$ or there exists $u' \in \Sigma^\infty$ such that $uu' = w$.


Of particular importance are \emph{morphic words}~\cite{allo:shal:2003}.
For example:
\begin{enumerate}\setlength{\itemsep}{.5ex}
  \item 
    The \thuemorse{} word $\morse$ 
    arises by starting from the word $0$, 
    as the limit of repeatedly applying the morphism
    $0 \mapsto 01$, $1 \mapsto 10$.
    We abbreviate this by: $\tup{0 \mid 0 \mapsto 01,\, 1 \mapsto 10}$.
    The first iterations are
    $0 \mapsto 01 \mapsto 0110 \mapsto \cdots$. 
  \item 
    The period-doubling word $\perioddoubling = \tup{0 \mid 0\mapsto 01,\, 1\mapsto 00}$.
  \item 
    The Mephisto Waltz word~$\mephisto = \tup{0 \mid 0\mapsto 001,\, 1\mapsto 110}$.
\end{enumerate}
\xreplace{%
  For a more formal metric definition of morphic words, see~\cite{salo:81,allo:shal:2003}.
}{%
  Formally, morphic words are defined as follows~\cite{salo:81,allo:shal:2003}.
  Let 
  $s \in \Sigma^*$ be a starting word,
  $h : \Sigma \to \Sigma^*$ a morphism, and 
  $c : \Sigma \to \Sigma$ a coding.
  If the limit $h^\omega(s) = \lim_{i\to\infty} h^i(s)$
  exists, then 
  $c(h^\omega(s))$ is called a \emph{morphic word}.
  Here the limit is taken with respect to the following metric $d$ on $\Sigma^{\infty}$:
  $d(u,u) = 0$ and $d(u,v) = 2^{-n}$ for all $u,v \in \Sigma^{\infty}$ with $u \ne v$, 
  where $n$ is the length of the longest common prefix of $u$ and $v$.
}

\xtend{\section{Finite-state Transducers}\label{sec:fst}}

A \emph{sequential finite-state transducer} (FST)~\cite{allo:shal:2003,saka:03}, 
a.k.a.\ \emph{deterministic generalised sequential machine (DGSM)},
is a finite automaton with input letters and finite output words along the edges.
A transducer reads the input word letter by letter,
in each step producing an output word and changing its state.
Then the output word is the concatenation of all the output words encountered along the edges.%

\begin{definition}\normalfont
  A \emph{sequential finite-state transducer} 
  $A = \sixtuple{\alphin}{\alphout}{\states}{\istate{0}}{\stransfun}{\soutfun}$
  consists of
  a finite \emph{input alphabet} $\alphin$,
  a finite \emph{output alphabet} $\alphout$,
  a finite set of \emph{states} $\states$,
  an \emph{initial state} $\istate{0} \in \states$,
  a \emph{transition function} ${\stransfun} : {\states \times \alphin \to \states}$, and
  an \emph{output function} ${\soutfun} : {\states \times \alphin \to \alphout^*}$.
  Whenever the alphabets $\alphin$ and $\alphout$ are clear from the context, 
  we write $A = \quadruple{\states}{\istate{0}}{\stransfun}{\soutfun}$.
\end{definition}

An example of an FST is depicted in Figure~\ref{fig:fst:diff},
where we write `$a \,\rpl\, w$' along the transitions to indicate that the input
letter is $a$ and the output word is $w$.
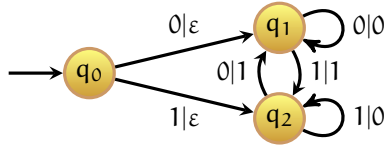
\begin{figure}[h!]
  \centering\vspace{-1ex}
  \begin{tikzpicture}[thick,>=stealth,very thick,node distance=25mm,
    node/.style={circle,gyellow,inner sep=.5ex,minimum size=6mm}]
    \node (start) {};
    \node (q0) [node,right of=start,node distance=12mm] {$q_{0}$};
    \node (q1) [node,right of=q0,yshift=6mm] {$q_{1}$};
    \node (q2) [node,right of=q0,yshift=-6mm] {$q_{2}$};
    \draw [->] (start) -- (q0);
    \draw [->] (q0) -- (q1) node [midway, above] {$0\rpl\wordemp$};
    \draw [->] (q0) -- (q2) node [midway, below] {$1\rpl\wordemp$};
    \draw [->] (q1) to[bend left=30] node [midway, right] {$1 \rpl 1$} (q2) ;
    \draw [->] (q2) to[bend left=30] node [midway, left] {$0 \rpl 1$} (q1);
    \draw [->,>=stealth'] (q2) .. controls ($(q2)+(10mm,-5mm)$) and ($(q2)+(10mm,5mm)$) .. (q2) node [midway, right] {$1 \rpl 0$}; 
    \draw [<-,>=stealth'] (q1) .. controls ($(q1)+(10mm,-5mm)$) and ($(q1)+(10mm,5mm)$) .. (q1) node [midway, right] {$0 \rpl 0$}; 
  \end{tikzpicture}\vspace{-1ex}
  \vspace{-.2cm}
  \caption{An FST realising the difference of consecutive bits modulo~2.
    For example, $\morse = 01101001\cdots$ 
    is transformed in $\overline{\perioddoubling} = 1011101\cdots$
    where the overbar signifies inversion between $0$ and $1$.
  }
  \label{fig:fst:diff}
\end{figure}

The output given by a transition is allowed to be a \emph{word} over the output alphabet,
and not just a single letter or the empty word $\wordemp$, although that may also be the case.
Thereby finite-state transducers generalize the class of Mealy machines
that output precisely one letter in each step.

\xtend{%
  The transducer in Figure~\ref{fig:fst:diff} 
  computes the first difference of the input word.
  For example, it reduces the \thuemorse{} sequence $\morse$ to
  the inverted period doubling sequence $\overline{\perioddoubling}$: 
  \begin{align*}
   &&&&&& &\ 0\;\;\;1\;\;\;1\;\;\;0\;\;\;1\;\;\;0\;\;\;0\;\;\;1\;\;\ldots &&= \morse &&&&&&\\[-.5ex]
   &&&&&& \to&\ \;\;\;1\;\;\;0\;\;\;1\;\;\;1\;\;\;1\;\;\;0\;\;\;1\;\;\ldots &&= \overline{\perioddoubling}
  \end{align*}
  Formally, the transduction of words is defined as follows.
}%
We only consider sequential transducers and will simply speak of finite-state transducers henceforth.
\begin{definition}\normalfont
  Let $A = \sixtuple{\alphin}{\alphout}{\states}{\istate{0}}{\stransfun}{\soutfun}$
  be a finite-state transducer.
  We homomorphically extend the transition function $\delta$ to $Q \times \alphin^* \to Q$ as follows:
  for $q \in Q$, $a \in \alphin$, $u \in \alphin^*$ let
  $\delta(q,\emptyword) = q$ and $\delta(q,au) = \delta(\delta(q,a),u)$.
  We extend the output function~$\lambda$ to $Q \times \alphin^\infty \to \alphout^\infty$ as follows:
  for $q \in Q$, $a \in \alphin$, $u \in \alphin^\infty$, let
  $\lambda(q,\emptyword) = \emptyword$ and $\lambda(q,au) = \lambda(q,a) \cdot \lambda(\delta(q,a),u)$.
\end{definition}
\xtend{%
  We note that finite-state transducers can be viewed as productive term rewrite systems~\cite{endr:hend:2011}
  and the transduction of infinite words as infinitary rewriting~\cite{endr:hans:hend:polo:silv:2015}.
}

\section{Transducer Degrees}\label{sec:degrees}

We now explain how FSTs
give rise to a hierarchy of degrees of infinite words, called transducer degrees.
First, we formally introduce the
transducibility relation $\ge$ on words as realised by FSTs.

\begin{definition}\normalfont
  Let $\astr \in \str{\Sigma}$, $\bstr \in \str{\Gamma}$ for finite alphabets $\Sigma$,~$\Gamma$.
  Let $A = \sixtuple{\alphin}{\alphout}{\states}{\istate{0}}{\stransfun}{\soutfun}$
  be a FST. 
  We write $\astr \red{A} \bstr$ if $\bstr = \outfun{q_0}{\astr}$.
  We write 
  $\astr \fstred \bstr$, and say that $u$ is a \emph{transduct} of $w$,
  if there exists a FST $A$ such that $\astr \red{A} \bstr$. 
\end{definition}

Note that the transducibility relation $\fstred$ is a pre-order.
It thus induces a partial order of `degrees', 
the equivalence classes with respect to $\red{} \cap \redi{}$.
We denote equivalence using $\fstconv$.
It is not difficult to see that every word over a finite alphabet
is equivalent to a word over the alphabet $\two = \{\,0,1\,\}$. 
Thus every degree contains a representative from $\str{\two}$. 
For the study of transducer degrees it suffices therefore to consider words over the latter alphabet.

\begin{definition}\normalfont
  Define the equivalence relation 
  $\fstconv \;=\; (\red{} \cap \redi{})$.
  The \emph{(transducer) degree} $\convclass{\astr}$ of an infinite word $\astr$ 
  is the equivalence class of~$\astr$ with respect to $\fstconv$,
  that is, $\convclass{\astr} = \{\bstr \in \words \mid \astr \fstconv \bstr\}$.
%
  We write $\convclassset{\words}$ to denote the set of degrees 
  $\{\convclass{\astr} \mid \astr \in \str{\two}\}$.

  The \emph{transducer degrees} form the partial order $\pair{\convclassset{\words}}{\fstred}$\xtend{%
    \footnote{%
    We note that finite state transducers transform infinite words to finite or infinite words.
    The result of the transformation is finite
    if the transducer outputs the empty word~$\wordemp$ for all except a finite number of letters of the input word.
    We are interested in infinite words only, since the set of finite words would merely entail
    two spurious extra sub-bottom degrees in the hierarchy of transducer degrees.
    }%
  },
  induced by the pre-order $\fstred$ on $\words$, that is, for words $\astr,\bstr \in \words$ we have
  $\convclass{\astr} \fstred \convclass{\bstr} \iff \astr \fstred \bstr$.
\end{definition}

The \emph{bottom degree $\botdeg$} is 
the least degree of the hierarchy,
that is, the unique degree $\deg{a} \in \convclassset{\words}$ 
such that $\deg{a} \fstredi \deg{b}$ for every $\deg{b} \in \convclassset{\words}$;
it
consists of the ultimately periodic words,
that is, words of the form $uvvv\cdots$ for finite words $u,v$ where $v \ne \emptyword$.

\begin{definition}\normalfont
  An \emph{atom} is a minimal non-bottom degree,
  that~is, a degree $\deg{a} \in \convclassset{\words}$ such that 
  $\botdeg \fstredistrict \deg{a}$ and 
  there exists no $\deg{b} \in \convclassset{\words}$ 
  with $\botdeg \fstredistrict \deg{b} \fstredistrict \deg{a}$.
\end{definition} 


Although FSTs are very simple and elegant devices,
we hardly understand their power for transforming words~\cite{endr:klop:saar:whit:2015}.
No methods are available to answer simple questions such as:
\newcommand{\xx}[1]{\underline{\hphantom{#1}}}
\begin{enumerate}\setlength{\itemsep}{.5ex}
  \item [\myquestion] \qlabel{q:perioddoubling}
    Consider the period-doubling sequence $\perioddoubling$ and drop every third element%
    \xreplace{,
      resulting in $w = 01\xx{0}0 \;0\xx{1}01 \;\xx{0}10\xx{0} \;01\xx{0}0 \;0\xx{1}00\cdots$.
    }{:
    \begin{samepage}
    \begin{align*}
      \perioddoubling &= 0100 \;0101 \;0100 \;0100 \;0100\cdots\\
      w &= 01\xx{0}0 \;0\xx{1}01 \;\xx{0}10\xx{0} \;01\xx{0}0 \;0\xx{1}00\cdots
    \end{align*}
    \end{samepage}
    }
    Obviously we have $\perioddoubling \fstred w$. 
    Do we have $w \fstred \perioddoubling$?
    
  \item [\myquestion] \qlabel{q:mephisto}
    Are the degrees of \thuemorse{}~$\morse$ and Mephisto Waltz~$\mephisto$ incomparable?
\end{enumerate}

\section{Spiralling Words}\label{sec:tools}

We now consider \emph{spiralling words} over the alphabet $\two = \{0,1\}$
for which the distance of consecutive $1$'s in the word grows to infinity.
We additionally require that the sequence of distances between consecutive $1$'s
is ultimately periodic modulo every natural number.
The class of spiralling words
permits a characterisation of their transducts
in terms of weighted products.

For a function $f : \nat \to \nat$, we define $\seq{f} \in \str{\two}$ by
\begin{align*}
  \seq{f} = \textstyle\prod_{i = 0}^{\infty} 10^{f(i)} = 1 0^{f(0)} \, 1 0^{f(1)} \, 1 0^{f(2)} \, \cdots\,\;.
\end{align*}%
We write $\seq{f(n)}$ as shorthand for~$\seq{n \mapsto f(n)}$.

\begin{example}\label{ex:fuse}
  As an example of a transduction between sporadic words, 
  to get a feeling of what finite-state transducers can do on such words, consider
  \begin{math}
    \seq{n^3} \red{A} \seq{(2n)^3 + (2n+1)^3}
  \end{math}.
  Here the transducer $A$ removes the $1$ between the
  appropriate consecutive blocks of $0$'s, as in:
  \(
    1\, 10\, 10^8\, 10^{27}\, 10^{64}\, 10^{125}\ 1 \cdots \fstred 10\, 1 0^{(8+27)}\, 1 0^{(64+125)}\, 1 \cdots
  \). 
  It is easy to determine the two-state transducer $A$ that removes the $1$'s
  at the right places.
\end{example}

\begin{definition}\label{def:spiralling}\normalfont
  A function $f : \nat \to \nat$ 
  is called \emph{spiralling}
  if 
  \begin{enumerate}
    \item{}\label{def:spiralling:item:i}
      $\lim_{n \to \infty} f(n) = \infty$, and
    \item{}\label{def:spiralling:item:ii} 
      for every $m \ge 1$, the function $n \mapsto f(n) \bmod m$ is ultimately periodic.
  \end{enumerate}
  A word $\seq{f}$ is called \emph{spiralling} whenever $f$ is spiralling.
\end{definition}
\noindent
For example, $\seq{p(n)}$ is spiralling for every polynomial $p(n)$ with natural numbers as coefficients.\footnote{%
  The identity function and constants functions are spiralling.
  Moreover, the class of spiralling functions is closed under addition and multiplication.
  From this it follows that polynomials with natural numbers as coefficients are spiralling.
}
Spiralling functions are called `cyclically ultimately periodic' in the literature~\cite{bers:boas:cart:peta:pin:2006}. 
For a tuple $\vec{\alpha} = \tup{\alpha_0,\ldots,\alpha_{m}}$, we define
\begin{itemize}
  \item the \emph{length} $|\vec{\alpha}| = m+1$, and
  \item its \emph{rotation} by $\vec{\alpha}' = \tup{\alpha_1,\ldots ,\alpha_{m},\alpha_0}$.
\end{itemize}
Let $A$ be a set and $f : \nat \to A$ a function.
We write $\shift{k}{f}$ for the \emph{$k$-th shift} of $f$ 
defined by $\shift{k}{f}(n) = f(n+k)$.

We use `weights' to represent linear functions.
\begin{definition}\label{def:weight}\normalfont
  A \emph{weight} $\vec{\alpha}$ is a tuple $\tup{a_0,\ldots,a_{k-1},b} \in \rat^{k+1}$ 
  of rational numbers such that $k \in \nat$ and $a_0,\ldots,a_{k-1} \ge 0$.
  The weight $\vec{\alpha}$ is called 
  \begin{itemize}
    \item \emph{non-constant} if $a_i \ne 0$ for some $i < k$, else \emph{constant},
    \item \emph{strongly non-constant} if $a_i,a_j \ne 0$ for some $i < j < k$.
  \end{itemize}
  Now let us also consider a tuple of tuples.
  For a tuple $\vec{\alpha} = \tup{\vec{\alpha_0},\ldots,\vec{\alpha_{m-1}}}$ of weights 
  we define 
  $
    \sumlength{\vec{\alpha}} = \textstyle\sum_{i = 0}^{m-1} (\,\length{\vec{\alpha_i}} - 1\,) \;.
  $
\end{definition}

\begin{definition}\label{def:weighted:product}\normalfont
  Let $f : \nat \to \rat$ be a function.
  For a weight~$\vec{\alpha} = \tup{a_0,\ldots,a_{k-1},b}$ 
  we define
  $\wof{\vec{\alpha}}{f} \in \rat$ by
  \,$
    \wof{\vec{\alpha}}{f} \;=\; a_0 f(0) + a_1 f(1) + \cdots + a_{k-1} f(k-1) + b \,.
  $ 
  For a tuple of weights $\vec{\alpha} = \tup{\vec{\alpha_0},\vec{\alpha_1},\ldots, \vec{\alpha_{m-1}}}$,
  we define the \emph{weighted product} 
  $\wprod{\vec{\alpha}}{f} : \nat \to \rat$
  by induction on $n$: 
  \begin{align*}
    (\wprod{\vec{\alpha}}{f})(0) & = \vec{\alpha_0} \cdot f \\
    (\wprod{\vec{\alpha}}{f})(n+1) & = (\wprod{\vec{\alpha}'}{\shift{|\vec{\alpha_0}|-1}{f}})(n) && (n \in \nat)
  \end{align*}
  We say that $\wprod{\vec{\alpha}}{f}$ is a \emph{natural weighted product} 
  if $(\wprod{\vec{\alpha}}{f})(n) \in \nat$ for all $n \in \nat$. 
\end{definition}

Weighted products are easiest understood by examples.
\begin{example}\label{ex:wprod}
  Let $f(n) = n^2$ be a function and 
  $\vec{\alpha} = \tup{\vec{\alpha_0},\vec{\alpha_1}}$ a tuple of weights 
  with $\vec{\alpha_0} = \tup{1,2,3,4}$, $\vec{\alpha_1} = \tup{0,1,1}$. 
  Then the weighted product $\wprod{\vec{\alpha}}{f}$ can be visualised as follows
  \begin{center}\vspace{-0.5ex}
    \begin{tikzpicture}[scale=0.95,nodes={scale=1}]
      \node at (-.5cm,0) [anchor=east] {$f$};
      \node at (10.1*.8cm,0) [anchor=east] {$\cdots$};
      \foreach \i/\j in {0/0,1/1,2/4,3/9,4/16,5/25,6/36,7/49,8/64,9/81} {
        \node (\i) at (\i*.8cm,0) {\j};
      }
      \node at (-.2cm,-1.2cm) [anchor=east] {$\wprod{\vec{\alpha}}{f}$};
      \node at (10.1*.8cm,-1.2cm) [anchor=east] {$\cdots$};
      \foreach \i/\x/\v in {0/1/18,1/3.5/17,2/6/248,3/8.5/82} {
        \node (s\i) at (\x*.8cm,-1.2cm) {\v};
      }
      \begin{scope}[inner sep=0,->,nodes={scale=.8}]
      \draw (0) -- node [xshift=3mm,pos=.3] {$\times1$} (s0); 
      \draw (1) -- node [xshift=2.5mm,pos=.3] {$\times2$} (s0); 
      \draw (2) -- node [xshift=3mm,pos=.3] {$\times3$} node [xshift=3mm,pos=.8] {$+4$} (s0); 
      \draw (3) -- node [xshift=3mm,pos=.3] {$\times0$} (s1); 
      \draw (4) -- node [xshift=3mm,pos=.3] {$\times1$} node [xshift=3mm,pos=.8] {$+1$} (s1); 
      \draw (5) -- node [xshift=3mm,pos=.3] {$\times1$} (s2); 
      \draw (6) -- node [xshift=2.5mm,pos=.3] {$\times2$} (s2); 
      \draw (7) -- node [xshift=3mm,pos=.3] {$\times3$} node [xshift=3mm,pos=.8] {$+4$} (s2); 
      \draw (8) -- node [xshift=3mm,pos=.3] {$\times0$} (s3); 
      \draw (9) -- node [xshift=3mm,pos=.3] {$\times1$} node [xshift=3mm,pos=.8] {$+1$} (s3); 
      \end{scope}
    \end{tikzpicture}
  \end{center}
  Intuitively, the weight $\vec{\alpha_0} = \tup{1,2,3,4}$ 
  means that $3$ consecutive entries are added while 
  being multiplied by $1$, $2$ and $3$, respectively,
  and $4$ is added to the result.
\end{example}

We introduce a few operations on weights.
We define scalar multiplication of weights in the obvious way.
We also introduce a multiplication $\odot$ that affects only the last entry of weights
(the constant term).
\begin{definition}\normalfont
  Let $c \in \rat_{\ge 0}$, $\vec{\alpha} = \tup{a_{0},\ldots,a_{\ell-1},b}$ a weight,
  $\vec{\beta} = \tup{\vec{\beta_0},\ldots,\vec{\beta_{m-1}}}$ a tuple of weights.
  We define 
  \begin{align*}
    c\vec{\alpha} &= \tup{ca_0,\ldots,ca_{\ell-1},cb} &
    \vec{\alpha} \odot c &= \tup{a_0,\ldots,a_{\ell-1},bc} \\
    c\vec{\beta} &= \tup{c\vec{\beta_0},\ldots,c\vec{\beta_{m-1}}} &
    \vec{\beta} \odot c &= \tup{\vec{\beta_0} \odot c,\ldots,\vec{\beta_{m-1}} \odot c}
  \end{align*}
%
\end{definition}

The next lemma follows directly from the definitions. 
\begin{lemma}\label{lem:wprod:mul}
  Let $c \in \rat_{\ge 0}$, 
  $\vec{\alpha}$ 
  a tuple of weights,
  and $f : \nat \to \rat$ a function.
  Then 
  $c(\wprod{\vec{\alpha}}{f}) = \wprod{(c\vec{\alpha})}{f} = \wprod{(\vec{\alpha} \odot c)}{(cf)}$.
  \qed
\end{lemma}

It is straightforward to define a \emph{composition} of tuples of weights
such that 
$\wprod{\vec{\beta}}{(\wprod{\vec{\alpha}}{f})} = \wprod{(\wprod{\vec{\beta}}{\vec{\alpha}})}{f}$
for every function $f : \nat \to \rat$.
Note that $\wprod{\vec{\alpha}}{f}$ is already defined.
For the precise definition of $\wprod{\vec{\beta}}{\vec{\alpha}}$, 
we refer \xreplace{to~\cite{arxiv}}{to Appendix~\ref{sec:composition}}.
We will employ the following two properties of composition.
\begin{lemma}\label{lem:wprod:compose}
  Let $\vec{\alpha},\vec{\beta}$ be tuples of weights.
  Then we have that
  \begin{math}
    \wprod{\vec{\beta}}{(\wprod{\vec{\alpha}}{f})} 
    = \wprod{(\wprod{\vec{\beta}}{\vec{\alpha}})}{f}
  \end{math}
  for every function $f : \nat \to \rat$.
  \qed
\end{lemma}

\begin{lemma}\label{lem:wprod:double:non:constant}
  Let $\vec{\alpha}$ be tuple of weights, and $\vec{\beta}$ a tuple of strongly non-constant weights.
  Then $\wprod{\vec{\alpha}}{\vec{\beta}}$ is of the form $\tup{\gamma_0,\ldots,\gamma_{k-1}}$
  such that for every $i \in \nat_{<k}$, the weight~$\gamma_i$ is either constant or strongly non-constant.
  \qed
\end{lemma}



We need a few results on weighted products from~\cite{endr:grab:hend:zant:2015}. 
The following lemma states that
every natural weighted product (see Definition~\ref{def:weighted:product}) can be realised by a FST.
\begin{lemma}[\cite{endr:grab:hend:zant:2015}]\label{lem:wprod:FST}
  Let $f : \nat \to \nat$, and 
  $\vec{\alpha}$ a tuple of weights.
  If ${\wprod{\vec{\alpha}}{f}}$ is a natural weighted product (i.e., $\forall n \in \nat.\; (\wprod{\vec{\alpha}}{f})(n) \in \nat$), then 
  $\seq{f} \fstred \seq{\wprod{\vec{\alpha}}{f}}$.
  \qed
\end{lemma}

For the proof of Theorem~\ref{thm:transducts}, below, 
we use the following auxiliary lemma.
The lemma gives a detailed structural analysis, elaborated and explained in~\cite{endr:grab:hend:zant:2015},
of the transducts of a spiralling word $\seq{f}$.
\begin{lemma}[\cite{endr:grab:hend:zant:2015}]\label{lem:disambiguate}
  Let $f : \nat \to \nat$ be a spiralling function, 
  and let $\sigma \in \str{\two}$ be 
  such that $\seq{f} \fstred \sigma$ and $\sigma \not\in\botdeg$.
  Then there exist 
  $n_0, m \in \nat$, a word $w \in \two^*$,
  a tuple of weights~$\vec{\alpha}$,
  and tuples of finite words $\vec{p}$ and $\vec{c}$ 
  with $\length{\vec{\alpha}} = \length{\vec{p}} = \length{\vec{c}} = m > 0$
  such that
  $
    \sigma = w \cdot \prod_{i = 0}^{\infty} \prod_{j = 0}^{m-1} p_j \, c_j^{\cyc{i}{j}}
  $ 
  where $\cyc{i}{j} = (\wprod{\vec{\alpha}}{\shift{n_0}{f}})(mi + j)$, and
  \begin{enumerate}
    \item 
      $c_j^\omega \ne p_{j+1} c_{j+1}^\omega$ for every $j$ with $0 \le j < m-1$, 
      and $c_{m-1}^\omega \ne p_{0} c_{0}^\omega$, and
      \label{item:no:ambiguities}
    \item 
      $c_j \ne \emptyword$,
      and $\alpha_j$ is non-constant, for all $j \in \nat_{<m}$.
      \label{item:no:empty:cycles}
  \qed
  \end{enumerate}
\end{lemma}

\begin{example}
  We continue Example~\ref{ex:wprod}.
  We have $\vec{\alpha} = \tup{\vec{\alpha_0},\vec{\alpha_1}}$.
  Accordingly, we have prefixes $p_0,p_1 \in \two^*$ and cycles $c_0,c_1 \in \two^*$.
  Then the transduct $\sigma$ in Lemma~\ref{lem:disambiguate},
  defined by the double product, 
  can be derived as follows:
  \begin{center}\vspace{-.5ex}
    \begin{tikzpicture}[scale=0.95,nodes={scale=1}]
      \node at (-.5cm,0) [anchor=east] {$f$};
      \node at (10.1*.8cm,0) [anchor=east] {$\cdots$};
      \foreach \i/\j in {0/0,1/1,2/4,3/9,4/16,5/25,6/36,7/49,8/64,9/81} {
        \node (\i) at (\i*.8cm,0) {\j};
      }
      \node at (-.2cm,-0.8cm) [anchor=east] {$\wprod{\vec{\alpha}}{f}$};
      \node at (10.1*.8cm,-0.8cm) [anchor=east] {$\cdots$};
      \foreach \i/\x/\v in {0/1/18,1/3.5/17,2/6/248,3/8.5/82} {
        \node (s\i) at (\x*.8cm,-0.8cm) {\v};
      }
      \begin{scope}[inner sep=0,nodes={scale=.8}]
      \draw (0.south west) -- (s0); 
      \draw (2.south east) -- (s0); 
      \draw (3.south west) -- (s1); 
      \draw (4.south east) -- (s1); 
      \draw (5.south west) -- (s2); 
      \draw (7.south east) -- (s2); 
      \draw (8.south west) -- (s3); 
      \draw (9.south east) -- (s3);
      \end{scope}

      \node at (1*0.8,-.4) {$\vec{\alpha_0}$};
      \node at (3.5*0.8,-.4) {$\vec{\alpha_1}$};
      \node at (6*0.8,-.4) {$\vec{\alpha_0}$};
      \node at (8.5*0.8,-.4) {$\vec{\alpha_1}$};

      \node at (-.2cm,-1.4cm) [anchor=east] {$\sigma = w$};
      \node at (10.1*.8cm,-1.4cm) [anchor=east] {$\cdots$};
      \node at (0,-1.4) {$\cdot$};
      \node at (1*0.8,-1.4) {$p_0\, c_0^{18}$};
      \node at (2.15*0.8,-1.4) {$\cdot$};
      \node at (3.5*0.8,-1.4) {$p_1\, c_1^{17}$};
      \node at (4.65*0.8,-1.4) {$\cdot$};
      \node at (6*0.8,-1.4) {$p_0\, c_0^{248}$};
      \node at (7.15*0.8,-1.4) {$\cdot$};
      \node at (8.5*0.8,-1.4) {$p_1\, c_1^{82}$};
    \end{tikzpicture}\vspace{-.5ex}
  \end{center}
  The infinite word $\sigma$ is the 
  infinite concatenation of $w$
  followed by alternating $p_0 c_0^{e_0}$ and $p_1 c_1^{e_1}$,  
  where the exponents $e_0$ and $e_1$ are the result of applying weights $\vec{\alpha_0}$ and $\vec{\alpha_1}$, respectively.
\end{example}

We characterise the transducts of spiralling words up to equivalence ($\fstconv$).
\begin{theorem}[\cite{endr:grab:hend:zant:2015}]\label{thm:transducts:up:to}
  Let $f : \nat \to \nat$ be spiralling, and $\sigma \in \str{\two}$.
  Then $\seq{f} \fstred \sigma$
  if and only if 
  $\sigma \fstconv \seq{\wprod{\vec{\alpha}}{\shift{n_0}{f}}}$
  for some $n_0 \in \nat$, and a tuple of weights $\vec{\alpha}$.
\end{theorem}

Roughly speaking, polynomials of order $k$
are closed under transduction.
\begin{proposition}[\cite{endr:grab:hend:zant:2015}]\label{prop:poly-transducts}
  Let $p(n)$ be a polynomial of order $k$ with non-negative integer coefficients,
  and let $\sigma\notin\botdeg$ with $\seq{p(n)} \fstred \sigma$.
  Then $\sigma \fstred \seq{q(n)}$
  for some polynomial $q(n)$ of order~$k$
  with non-negative integer coefficients.
\end{proposition}

\section{The Degree of $\seq{n^k}$ is Not an Atom for $k \ge 3$}\label{sec:nk}

We show that the degree of $\seq{n^k}$ is not an atom for~$k \ge 3$.
For this purpose,
we prove a strengthening of Theorem~\ref{thm:transducts:up:to},
a lemma on weighted products of strongly non-constant weights,
and we employ the power mean inequality~\cite{hard:plya:1988}.

\begin{definition}\normalfont
  For $p \in \reals$, the \emph{weighted power mean} $M_p(\vec{x})$ of 
  $\vec{x} = \tup{x_1,x_2,\ldots,x_n} \in \reals_{>0}^n$ 
  with respect to $\vec{w} = \tup{w_1,w_2,\ldots,w_n} \in \reals_{>0}^n$
  with $\sum_{i = 1}^n w_i = 1$
  is  
  \begin{align*}
    M_{\vec{w},0}(\vec{x}) &= \textstyle\prod_{i=1}^{n} x_i^{w_i} &
    M_{\vec{w},p}(\vec{x}) &= (\textstyle\sum_{i=1}^{n} w_i x_i^p)^{1/p} \,.
  \end{align*}
\end{definition}

\begin{proposition}[Power mean inequality]\label{prop:weighted:mean}
  For all $p,q\in\reals$, $\vec{x},\vec{w} \in \reals_{>0}^n$:
  \begin{gather*}
    p < q \implies M_{\vec{w},p}(\vec{x}) \le M_{\vec{w},q}(\vec{x})\\
    ( p = q \vee x_1 = x_2 = \cdots = x_n ) \iff M_{\vec{w},p}(\vec{x}) = M_{\vec{w},q}(\vec{x}) \,.
  \end{gather*}
\end{proposition}

Theorem~\ref{thm:transducts:up:to} characterises transducts of spiralling sequences only up to equivalence.
This makes it difficult to employ the theorem for proving non-transducibility.
We improve the characterisation for the case of spiralling transducts as follows.
\begin{theorem}\label{thm:transducts}
  Let $f,g : \nat \to \nat$ be spiralling functions.
  Then $\seq{g} \fstred \seq{f}$ if and only if
  some shift of $f$ is a weighted product of a shift of $g$, that is:
  \begin{align*}
    \shift{n_0}{f} = \wprod{\vec{\alpha}}{\shift{m_0}{g}}
  \end{align*}
  for some $n_0, m_0 \in \nat$ and a tuple of weights $\vec{\alpha}$.
\end{theorem}

\begin{proof}
  For the direction `$\Leftarrow$', assume that $\shift{n_0}{f} = \wprod{\vec{\alpha}}{\shift{m_0}{g}}$.
  Then we have
  $\seq{g} \fstconv \seq{\shift{m_0}{g}} \fstred \seq{\wprod{\vec{\alpha}}{\shift{m_0}{g}}} = \seq{\shift{n_0}{f}} \fstconv \seq{f}$  
  by invariance under shifts and by Lemma~\ref{lem:wprod:FST}.

  For the direction `$\Rightarrow$', assume that $\seq{g} \fstred \seq{f}$.
  Then by Lemma~\ref{lem:disambiguate} there exist
  $m_0,m \in \nat$, $w \in \two^*$, $\vec{\alpha}$, $\vec{p}$ and $\vec{c}$ 
  with $\length{\vec{\alpha}} = \length{\vec{p}} = \length{\vec{c}} = m > 0$
  such that:
  \begin{align}
    \seq{f} = \textstyle w \cdot \prod_{i = 0}^{\infty} \prod_{j = 0}^{m-1} p_j \, c_j^{\cyc{i}{j}}
    \label{eq:f:product}
  \end{align}
  where $\cyc{i}{j} = (\wprod{\vec{\alpha}}{\shift{m_0}{g}})(mi + j)$
  such that the conditions~\ref{item:no:ambiguities} and~\ref{item:no:empty:cycles} of Lemma~\ref{lem:disambiguate}
  are fulfilled.
  
  Note that, as $\lim_{n\to \infty} f(n) = \infty$, the distance of ones in the sequence $\seq{g}$ tends to infinity.
  For every $j \in \nat_{<m}$,
  the word $p_j$ occurs infinitely often in $\seq{f}$ by~\eqref{eq:f:product},
  and hence $p_j$ can contain at most one occurrence of the symbol $1$.

  By condition~\ref{item:no:empty:cycles}, we have for every $j \in \nat_{<m}$ that $c_j \ne \emptyword$,
  and the weight $\vec{\alpha_j}$ is not constant.
  As $\lim_{n\to \infty} g(n) = \infty$, it follows that $c_j^2$ 
  appears infinitely often in $\seq{f}$ by~\eqref{eq:f:product}.
  Hence $c_j$ consists only of $0$'s, that is, $c_j \in \{0\}^+$ for every $j \in \nat_{<m}$.
  
  By condition~\ref{item:no:ambiguities} 
  we never have $c_j^\omega = p_{j+1} c_{j+1}^\omega$ for $j \in \nat_{<m}$ (where addition is modulo $m$).
  As $c_j^\omega = 0^\omega$ and $p_{j+1} 0^\omega = p_{j+1} c_{j+1}^\omega$,
  we obtain that $p_{j+1}$ must contain a $1$.
  Hence, for every $k \in \nat_{<m}$, the word $p_j$ contains precisely one $1$.

  Finally, we apply the following transformations 
  to ensure $p_j = 1$ and $c_j = 0$ for every $j \in \nat_{<m}$:
  \begin{enumerate}
    \item 
    For every $j \in \nat_{<m}$ such that $c_j = 0^h$ for some $h > 1$, 
    we set $c_j = 0$ and replace the weight $\vec{\alpha_j}$ in $\vec{\alpha}$ by $h\vec{\alpha_j}$.
    \item 
    For every $j \in \nat_{<m}$ such that $p_j = 0^h 1 0^\ell$ for some $h \ge 1$ or $\ell \ge 1$,
    we set $p_j = 1$ and replace the weight $\vec{\alpha_j}$ in $\vec{\alpha}$ by $(\vec{\alpha_j}+\ell)$ 
    and the weight $\vec{\alpha_{j-1}}$ by $(\vec{\alpha_{j-1}} + h)$.
    Here, for a weight $\vec{\gamma} = \tup{x_0,\ldots,x_{\ell-1},y}$ and $z \in \rat$, 
    we write $\vec{\gamma} + z$ for the weight $\tup{x_0,\ldots,x_{\ell-1},y+z}$.
    If $j = 0$, we moreover append $0^h$ to the word $w$.
  \end{enumerate}
  Both transformations preserve equation~\eqref{eq:f:product};
  the double product remains the same.
  
  Thus we now have $p_j = 1$ and $c_j = 0$ for every $j \in \nat_{<m}$.
  It follows from~\eqref{eq:f:product} that 
  $\seq{f} = w \seq{\wprod{\vec{\alpha}}{\shift{m_0}{g}}}$.
  Hence $\shift{n_0}{f} = \wprod{\vec{\alpha}}{\shift{m_0}{g}}$
  for some $n_0\in\nat$.  
\end{proof}

Theorem~\ref{thm:transducts} strengthens Theorem~\ref{thm:transducts:up:to}
in the sense that the characterisation uses equality ($=$ and shifts) instead of equivalence ($\fstconv$).
But Theorem~\ref{thm:transducts} only characterises spiralling transducts whereas Theorem~\ref{thm:transducts:up:to} characterises all transducts. 
However, next we will employ the gained precision to show that 
certain spiralling transducts of $\seq{n^k}$ cannot be transduced back to $\seq{n^k}$,
and conclude that $\seq{n^k}$ is not an atom for \mbox{$k \ge 3$}.

\begin{theorem}\label{thm:not:atom}
  For $k \ge 3$, the degree of $\seq{n^k}$ is not an atom.
\end{theorem}

\begin{proof}
  Define $f : \nat \to \nat$ by $f(n) = n^k$. 
  We have $\seq{f} \fstred \seq{g}$ where $g : \nat \to \nat$ is defined by $g(n) = (2n)^k + (2n+1)^k$;
  cf.\ Example~\ref{ex:fuse}.
  Assume that we had $\seq{g} \fstred \seq{f}$.
  Then, by Theorem~\ref{thm:transducts} we have
  $\shift{n_0}{f} = \wprod{\vec{\alpha}}{\shift{m_0}{g}}$
  for some
  $n_0, m_0 \in \nat$ and a tuple of weights $\vec{\alpha}$.
  Note that $g = \wprod{\tup{\tup{1,1,0}}}{f}$
  and 
  \begin{align*}
    \shift{n_0}{f} 
    &= \wprod{\vec{\alpha}}{\shift{m_0}{\wprod{\tup{\tup{1,1,0}}}{f}}} \\
    &= \wprod{\vec{\alpha}}{(\wprod{\tup{\tup{1,1,0}}}{\shift{2 m_0}{f}})} 
    = \wprod{\vec{\beta}}{\shift{2 m_0}{f}}
  \end{align*}
  where $\vec{\beta} = \wprod{\vec{\alpha}}{\tup{\tup{1,1,0}}}$.
  By Lemma~\ref{lem:wprod:double:non:constant} 
  every weight in $\vec{\beta}$ is either
  constant or strongly non-constant.
  As $\shift{n_0}{f}$ is strictly increasing (and hence contains no constant subsequence),
  each weight in $\vec{\beta}$ must be strongly non-constant.

  Let $\vec{\beta} = \tup{\vec{\beta_0},\ldots,\vec{\beta_{\ell-1}}}$.
  For every $n \in \nat$ we have
  \begin{align}
    \begin{aligned}
    \shift{n_0}{f}(\ell n) 
    &= (\wprod{\vec{\beta}}{\shift{2 m_0}{f}})(\ell n) 
    = \wof{\vec{\beta_0}}{\shift{2 m_0 + \sumlength{\vec{\beta}}\cdot n}{f}}\;.
    \end{aligned}
    \label{eq:not:prime:eq}
  \end{align}
  Then we have
  \begin{align} 
    \shift{n_0}{f}(\ell n) &= (n_0 + \ell n)^k \notag
                      = \textstyle\sum_{i=0}^{k} \binom{k}{i} n_0^i \ell^{k-i}n^{k-i} \notag\\
                      &= \ell^k n^k + k n_0 \ell^{k-1} n^{k-1} + \cdots + k n_0^{k-1} \ell n + n_0^k 
                      \;. \label{eq:not:prime:left}
  \end{align}
  Let $\vec{\beta_0} = \tup{a_0,a_1,\ldots,a_{h-1},b}$.
  We define $c_i = a_i \sumlength{\vec{\beta}}^k$ and $d_i = (2m_0 + i)/ \sumlength{\vec{\beta}}$. 
  We obtain
  \begin{align}
    \ \wof{\vec{\beta_0}}{\shift{2 m_0 + \sumlength{\vec{\beta}}\cdot n}{f}} \notag 
    =&\ b + \textstyle\sum_{i = 0}^{h-1} a_i f(2 m_0 + \sumlength{\vec{\beta}}\cdot n + i) \notag \\
    =&\ b + \textstyle\sum_{i = 0}^{h-1} a_i f(\sumlength{\vec{\beta}}(n + \frac{2m_0+i}{\sumlength{\vec{\beta}}})) \notag \\
    =&\ b + \textstyle\sum_{i = 0}^{h-1} a_i \sumlength{\vec{\beta}}^k(n + d_i)^k = b + \sum_{i = 0}^{h-1} c_i (n + d_i)^k \notag \\
    =&\ b + \textstyle\sum_{i = 0}^{h-1} c_i (n^k + k d_i n^{k-1} + \cdots + k d_i^{k-1} n + d_i^k) \;. \label{eq:not:prime:right}
  \end{align}
  Recall equation~\eqref{eq:not:prime:eq}.
  Comparing the coefficients of $n^k$, $n^{k-1}$ and~$n$ in \eqref{eq:not:prime:left} and \eqref{eq:not:prime:right} 
  we obtain
  \begin{align*}
    \ell^k &= \sum_{i = 0}^{h-1} c_i 
    &
    kn_0 \ell^{k-1} &= \sum_{i = 0}^{h-1} c_i k d_i 
    &
    kn_0^{k-1} \ell &= \sum_{i = 0}^{h-1} c_i k d_i^{k-1} \,\text{, and hence} \\
    1 &= \sum_{i = 0}^{h-1} \frac{c_i}{\ell^k} 
    &
    \frac{n_0}{\ell} &= \sum_{i = 0}^{h-1} \frac{c_i}{\ell^{k}} d_i 
    &
    \frac{n_0^{k-1}}{\ell^{k-1}} &= \sum_{i = 0}^{h-1} \frac{c_i}{\ell^{k}} d_i^{k-1} \,,
  \end{align*}
  contradicting the weighted power means inequality (Proposition~\ref{prop:weighted:mean}).
  Clearly all $d_i$ are distinct, and, as a consequence of $\vec{\beta_0}$ being strongly non-constant,
  there are at least two $i \in \nat_{<h}$ for which $c_i \ne 0$.
  Thus our assumption $\seq{g} \fstred \seq{f}$ is wrong.
  Hence the degree of $\seq{n^k}$ is not an atom.
\end{proof}

\section{Atoms of Every Polynomial Order}\label{sec:polynomial}

\newcommand{\qpoly}{\mathfrak{Q}}
\newcommand{\npoly}{\mathfrak{N}}

The previous section stated that $\seq{n^k}$ is not an atom for $k \ge 3$.
Now we show that for every $k \in \nat$ there exists a polynomial $p(n)$ of order $k$
such that the degree of the word~$\seq{p(n)}$ is an atom.
Hence there are at least $\aleph_0$ atoms in the transducer degrees.

As we have seen in the proof of Theorem~\ref{thm:not:atom}, whenever $k \ge 3$,
we have that $\seq{n^k} \fstred \seq{g(n)}$,
but not $\seq{g(n)} \fstred \seq{n^k}$ for $g(n) = (2n)^k + (2n+1)^k$.
Thus there exist polynomials $p(n)$ of order $k$
for which $\seq{p(n)}$ cannot be transduced to $\seq{n^k}$.
However, the key observation underlying the construction in this section is the following:
Although we may not be able to reach $\seq{n^k}$ from $\seq{p(n)}$,
we can get arbitrarily close (Lemma~\ref{lem:epsilon}, below).
This enables us to employ the concept of \emph{continuity}.

In order to have continuous functions over the space of polynomials to allow limit constructions,
we now permit rational coefficients. 
For $k \in \nat$, let $\qpoly_k$ be the set of polynomials of order $k$ 
with non-negative rational coefficients.
We also use polynomials in $\qpoly_k$ to denote spiralling sequences.
However, we need to give meaning to $\seq{q(n)}$ for the case that the block sizes $q(n)$ are not natural numbers.
For this purpose, we make use of the fact that the degree of a word $\seq{f(n)}$
is invariant under multiplication of the block sizes by a constant, as is easy to see.
More precisely, for $f : \nat \to \nat$,
we have $\seq{f(n)} \fstconv \seq{d \cdot f(n)}$ for every $d \in \nat$ with $d \ge 1$.
So to give meaning to $\seq{q(n)}$, 
we multiply the polynomial by the least natural number $d > 0$ such that
$d \cdot q(n)$ is a natural number for every $n\in \nat$.

\begin{definition}\label{def:seq:q}\normalfont
  We call a function $f : \nat \to \rat$ \emph{naturalisable}
  if there exists a natural number $d  \ge 1$ such that
  for all $n \in \nat$ we have $(d \cdot f(n)) \in \nat$.

  For naturalisable $f : \nat \to \rat$ we define
  $\seq{f} \;=\; \seq{d\cdot f}$
  where $d \in \nat$ is the least number 
  such that $d \geq 1$ where 
  for all $n \in \nat$ we have $(d \cdot f(n)) \in \nat$.
  (Note that, for $f : \nat \to \nat$, $\seq{f(n)}$  has been defined in Section~\ref{sec:tools}.)
\end{definition}
\noindent
Observe that every $q(n) \in \qpoly_k$ is naturalisable (multiply by the least common denominator of the coefficients).
Also, naturalisable functions are preserved under weighted products.

Lemma~\ref{lem:wprod:FST} generalises as follows.
We no longer need to require that the weighted product is natural.
All weighted products of naturalisable functions can be realised by finite-state transducers.
\begin{lemma}\label{lem:wprod:FST:q}
  Let $f : \nat \to \rat$ be naturalisable, and $\vec{\alpha}$ a tuple of weights.
  Then $\wprod{\vec{\alpha}}{f}$ is naturalisable and $\seq{f} \fstred \seq{\wprod{\vec{\alpha}}{f}}$.
\end{lemma}
\begin{proof}
  Let $\vec{\alpha} = \tup{\vec{\alpha_0},\ldots,\vec{\alpha_{m-1}}}$ for some $m \ge 1$.
  Let $c \in \nat$ with $c \ge 1$ be minimal such that
  all entries of $c\vec{\alpha}$ are natural numbers.
  Let $d \in \nat$ with $d \ge 1$ be the least natural number such that
  $\myal{n \in \nat}{(d \cdot f(n)) \in \nat}$.
  
  Then we obtain $(\wprod{(dc\vec{\alpha})}{f})(n) \in \nat$ for ever $n\in\nat$.
  By the definition of weighted products it follows immediately that 
  $\wprod{(dc\vec{\alpha})}{f} = dc(\wprod{\vec{\alpha}}{f})$,
  and hence $\wprod{\vec{\alpha}}{f}$ is naturalisable.
  Let $e \in \nat$ with $e \ge 1$ be the least natural number such that
  $\myal{n \in \nat}{(e \cdot (\wprod{\vec{\alpha}}{f})(n)) \in \nat}$.
  
  We have the following transduction 
  \begin{align*}
    \seq{f} &= \seq{df} &&\text{by Definition~\ref{def:seq:q}}\\
    &\fstred \seq{\wprod{((c\vec{\alpha}) \odot d)}{(df)}} &&\text{by Lemma~\ref{lem:wprod:FST}}\\
    &= \seq{\wprod{(dc\vec{\alpha})}{f}} = \seq{dc(\wprod{\vec{\alpha}}{f})} &&\text{by Lemma~\ref{lem:wprod:mul}}\\
    &\fstred \seq{\wprod{\tup{\tup{\frac{e}{dc},0}}}{(dc(\wprod{\vec{\alpha}}{f}))}} &&\text{by Lemma~\ref{lem:wprod:FST}}\\
    &= \seq{e(\wprod{\vec{\alpha}}{f})} = \seq{\wprod{\vec{\alpha}}{f}} &&\text{by Definition~\ref{def:seq:q}}
  \end{align*}
  This concludes the proof.
\end{proof}

The following lemma states 
that every word $\seq{q(n)}$, for~a polynomial $q(n) \in \qpoly_k$ of order $k$,
can be transduced arbitrarily close to (but perhaps not equal to) $\seq{n^k}$.

\begin{lemma}\label{lem:epsilon}
  Let $k \ge 1$ and let $q(n) \in \qpoly_k$ be a polynomial of order~$k$.
  For every $\varepsilon > 0$ we have
  $
    \seq{q(n)} \fstred \seq{n^k + b_{k-1} n^{k-1} + \cdots + b_1 n}
  $ 
  for some rational coefficients $0 \le b_{k-1},\ldots,b_1 < \varepsilon$.
\end{lemma}

\begin{proof}
  Let $q(n) = a_kn^k + a_{k-1} n^{k-1} + \cdots + a_1 n + a_0$,
  and let $\varepsilon > 0$ be arbitrary.
  Then for every $d \in \nat$, we have
  \begin{align*}
    \seq{q(n)} &\fstred \seq{q(dn)} 
    \fstred \seq{\frac{q(dn)}{a_k d^k}} 
    = \seq{n^k + \frac{a_{k-1}}{a_k d}n^{k-1} + \cdots + \frac{a_{1}}{a_k d^{k-1}}n^1 + \frac{a_{0}}{a_k d^{k}} } \\
    &\fstred \seq{n^k + \frac{a_{k-1}}{a_k d}n^{k-1} + \ldots + \frac{a_{1}}{a_k d^{k-1}}n^1 } 
  \end{align*}
  The first transduction selects a subsequence of the blocks.
  The second transduction is a division of the size of each block
  (application of Lemma~\ref{lem:wprod:FST:q} with the weight $\tup{\tup{1/a_k d^k,0}}$).
  The last transduction amounts to removing a constant number of zeros from each block
  (application of Lemma~\ref{lem:wprod:FST:q} with the weight $\tup{\tup{1,-a_{0}/(a_k d^{k})}}$).
  The last polynomial in the transduction is of the desired form
  if $d \in \nat$ is chosen large enough.
\end{proof}

For polynomials $p(n) \in \qpoly_k$,
we want to express weighted products $\wprod{\tup{\vec{\alpha}}}{p}$
in terms of matrix products, as follows.
\begin{definition}\normalfont
  For weights $\vec{\alpha} = \tup{a_0, \dots, a_{k - 1}, b}$
  we define a column vector
  \begin{gather*}
      U(\vec{\alpha}) = (a_0, \dots, a_{k - 1})^T \,.
  \end{gather*}
\end{definition}

\begin{definition}\normalfont
  If
  \begin{math}
      p(n) = \sum_{i = 0}^k c_i n^i
  \end{math}
  is a polynomial of order~$k$,
  we define a column vector
  \begin{math}
      V(p(n)) = (c_1, \dots, c_k)^T
  \end{math}
  and a square matrix
  \begin{equation*}
      M(p(n)) = (V(p(kn + 0)), \ \dots ,\ V(p(kn + k - 1)))\,.
  \end{equation*}
  We also write $V(p)$ short for $V(p(n))$ and $M(p)$ for $M(p(n))$.
\end{definition}

We have omitted the constant term $c_0$ from the definition of $V(p)$.
Because for every $f : \nat \to \nat$ and $c \in \nat$
we have $\seq{f(n)} \fstconv \seq{f(n) + c}$.
These words are of the same degree because a FST can 
add (or remove) a constant number of symbols $0$ to (from) every block of $0$'s.
Similarly, $b$ was omitted from the definition of $U(\vec{\alpha})$.

\begin{example}
  Consider the polynomial $n^3$:
  \begin{equation*}
      V(n^3) = \begin{pmatrix}
          0 \\
          0 \\
          1
      \end{pmatrix}
      \qquad \text{and} \qquad
      M(n^3) = \begin{pmatrix}
          0 & 9 & 36 \\
          0 & 27 & 54 \\
          27 & 27 & 27
      \end{pmatrix}
  \end{equation*}
  where the column vectors of the matrix $M(n^3)$ are given by $V((3n)^3)$, $V((3n+1)^3)$ and $V((3n+2)^3)$.
\end{example}

\begin{lemma} \label{lem:matrix}
  Let $k \geq 1$.
  Let $\vec{\alpha} = \tup{a_0, \dots, a_{k - 1}, b}$ be a weight
  and $p(n) \in \qpoly_k$.
  Then
  $M(p) \mmul U(\vec{\alpha}) = V(\wprod{\tup{\vec{\alpha}}}{p})$.
\end{lemma}

\begin{proof}
  We calculate
  \begin{align*}
      M(p) \mmul U(\vec{\alpha}) &= \sum_{i = 0}^{k - 1} a_i V(p(kn + i))
      = V\big(\sum_{i = 0}^{k - 1} a_i p(kn + i)\big) \\
      &= V\big(\sum_{i = 0}^{k - 1} a_i p(kn + i) + b\big)
      = V(\wprod{\tup{\vec{\alpha}}}{p})\;,
  \end{align*}
  which proves the lemma.
\end{proof}

Let us take a closer look at the matrix $M(n^k)$.
The element in the $i$th row and $j$th column is
$M_{i,j} = \binom{k}{i} k^i (j - 1)^{k - i}$\;.
Dividing the $i$th row by $\binom{k}{i} k^i$ for each $i$
gives a Vandermonde-type matrix, which is invertible.
Thus also $M(n^k)$ is invertible.  

\begin{lemma}\label{lem:invert}
  For $k \ge 1$, $M(n^k)$ is invertible. \qed
\end{lemma}

\begin{theorem}\label{thm:atom}
    Let $k \geq 1$.
    Let $a_0, \dots, a_{k - 1}$ be positive rational numbers,
    \begin{math}
        \vec{\alpha} = \tup{a_0, \dots, a_{k - 1}, 0},
    \end{math}
    and
    \begin{align*}
        p(n) = (\wprod{\tup{\vec{\alpha}}}{n^k})(n)
        = \sum_{i = 0}^{k - 1} a_i (k n + i)^k.
    \end{align*}
    Then
    \begin{math}
        \seq{q(n)} \fstred \seq{p(n)}
    \end{math}
    for all $q(n) \in \qpoly_k$.
    Moreover, the degree $\convclass{\seq{p(n)}}$ is an atom.
    Note that the degree $\convclass{\seq{p(n)}}$ is 
    the infimum of all degrees of words $\seq{q(n)}$ with $q(n) \in \qpoly_k$.
\end{theorem}

\begin{proof}
    By Lemma \ref{lem:matrix},
    \begin{math}
        M(n^k) \mmul U(\vec{\alpha}) = V(p).
    \end{math}
    By Lemma~\ref{lem:invert}, $M(n^k)$ is invertible and we can write
    \begin{math}
        U(\vec{\alpha}) = M(n^k)^{-1} V(p).
    \end{math}
    By Lemma~\ref{lem:epsilon},
    for every $\varepsilon > 0$ there exists $q_\varepsilon\in \qpoly_k$ such that
    $\seq{q(n)} \fstred \seq{q_\varepsilon(n)}$ and
    \begin{equation*}
        q_\varepsilon(n) = n^k + b_{k-1} n^{k-1} + \dots + b_1 n
    \end{equation*}
    with $0 \le b_i \le \varepsilon$ for every $i \in \{1, \dots, k - 1\}$.
    We will show that if $\varepsilon$ is small enough, then
    \begin{math}
        \seq{q_\varepsilon(n)} \fstred \seq{p(n)}.
    \end{math}

    We have
    \begin{math}
        \lim_{\varepsilon \to 0} M(q_\varepsilon) = M(n^k).
    \end{math}
    As $\det(M(n^3)) \ne 0$ and the determinant function is continuous,
    also $\det(M(q_\varepsilon)) \ne 0$ for all sufficiently small $\varepsilon$.
    Then $M(q_\varepsilon)$ is invertible, and we define
    \begin{math}
        U_\varepsilon = M(q_\varepsilon)^{-1} V(p).
    \end{math}
    We would like to have $U_\varepsilon = U(\gamma)$ for some weight $\gamma$.
    This is not always possible,
    because some elements of $U_\varepsilon$ might be negative.
    However, by the continuity of matrix inverse and product,
    \begin{align*}
        \lim_{\varepsilon \to 0} U_\varepsilon
        &= \lim_{\varepsilon \to 0} (M(q_\varepsilon)^{-1} V(p)) 
        = (\lim_{\varepsilon \to 0} M(q_\varepsilon))^{-1} V(p) 
        = M(n^k)^{-1} V(p) = U(\vec{\alpha}).
    \end{align*}
    Since every element of $U(\vec{\alpha})$ is positive,
    we can fix a small enough $\varepsilon$ so that
    every element of $U_\varepsilon$ is positive.
    Then we have $U_\varepsilon = U(\gamma)$ for some weight $\gamma$.

    We have
    \begin{math}
        M(q_\varepsilon) \mmul U(\gamma) = V(\wprod{\tup{\gamma}}{q_\varepsilon})
    \end{math}
    by Lemma \ref{lem:matrix},
    and
    \begin{math}
        M(q_\varepsilon) \mmul U(\gamma) = V(p)
    \end{math}
    by the definition of $U_\varepsilon$.
    As a consequence
    \begin{math}
        (\wprod{\tup{\gamma}}{q_\varepsilon})(n) = p(n) + c
    \end{math}
    for some constant $c$.
    By Lemma~\ref{lem:wprod:FST:q},
    we obtain
    \begin{math}
        \seq{q_\varepsilon(n)} \fstred \seq{p(n)}.
    \end{math}

   It remains to show that the degree $\convclass{\seq{p(n)}}$ is an atom.
   Assume that $\seq{p(n)} \fstred w$ and $w \not\in \botdeg$.
   By Proposition~\ref{prop:poly-transducts} we have 
   $w \fstred \seq{q(n)}$ for some $q(n) \in \qpoly_k$.
   As shown above, $\seq{q(n)} \fstred \seq{p(n)}$, thus $w \fstred \seq{p(n)}$.
   Hence $\convclass{\seq{p(n)}}$ is an atom.
\end{proof}

\section{A Hereditary Uncomputable Degree}\label{sec:uncomputable}

We show that for any countable set $\deg{D}$ of transducer degrees
that does not contain the bottom degree,
there exists a degree~$\deg{z} \ne \botdeg$ that such that $\deg{z}{\downarrow}$ 
contains no degree from $\deg{D}$.
Here $\deg{z}{\downarrow}$ is the \emph{cone} of $\deg{z}$, that is, the set of degrees below $\deg{z}$:
\begin{align*}
  \deg{z}{\downarrow} \;=\; \{\, \deg{a} \mid \deg{z} \fstred \deg{a} \,\} \;.
\end{align*}
To this end, we will prove the following theorem.
\begin{theorem}\label{thm:independent}
  Let $\mathcal{S} \subseteq \two^\nat$ be a countable set of words
  that contains no ultimately periodic words.
  Then there exists a word $w \in \two^\nat$ that is not ultimately periodic and 
  none of the transducts~$u$ of $w$, $w \fstred u$,
  resides in $\mathcal{S}$.
\end{theorem}

Before proving this theorem, we mention a few corollaries.
\begin{corollary}\label{cor:uncomputable}
  There exists an uncomputable word $\uncomp \in \two^\nat$
  whose finite-state transducts are all uncomputable, ultimately periodic or finite.
\end{corollary}
\begin{proof}
  Follows from Theorem~\ref{thm:independent} 
  with $S$ the set of computable words that are not ultimately periodic.
\end{proof}

Theorem~\ref{thm:independent} and Corollary~\ref{cor:uncomputable}
have the following immediate implications for the hierarchy of transducer degrees.
\begin{corollary}\label{cor:independent:degrees}
  Let $\deg{D}$ be a countable set of transducer degrees
  not containing the bottom degree.
  Then there exists a degree~$\deg{z} \ne \botdeg$
  that has no degrees in $\deg{D}$ below itself, that is,
  $\deg{z}{\downarrow} \cap \deg{D} = \emptyset$.
\end{corollary}

The following result is somewhat reminiscent of 
the situation in Turing degrees where there
exists a set of incomparable degrees of size continuum.

\begin{corollary}\label{cor:cones}
  Let $\deg{C}$ be a countable set of degrees
  with pairwise almost disjoint cones, that is, for all $\deg{a}, \deg{b} \in \deg{C}$~with~$\deg{a} \ne \deg{b}$,
  we have $\deg{a}{\downarrow} \cap \deg{b}{\downarrow} = \{ \botdeg \}$.
  Then $\deg{C}$ can be extended to an uncountable set
  of degrees with pairwise almost disjoint cones.
\end{corollary}

\begin{proof}
  Let $\deg{C}'$ be a maximal extension of $\deg{C}$.
  If $\deg{C}'$ was countable, then by Corollary~\ref{cor:independent:degrees}
  it could be extended by a disjoint cone:
  take $\deg{D} = \{ \deg{b} \mid \deg{a} \in \deg{C}', \deg{a} \fstred \deg{b}\} \setminus \{\botdeg\}$.
  This contradicts maximality of $\deg{C}'$.
\end{proof}
\noindent
We do not know if `uncountable' can be replaced by continuum in the corollary.

We call a degree \emph{uncomputable} if it contains an uncomputable word.
Note that degrees cannot contain both computable and uncomputable words
since the set of computable words is closed under finite-state transduction.

\begin{corollary}\label{cor:uncomputable:degrees}
  There exists an uncomputable transducer degree $\convclass{\uncomp}$
  that has only uncomputable degrees and the bottom degree below itself.
\end{corollary}
 
For the proof of Theorem~\ref{thm:independent}
we introduce a few auxiliary definitions and lemmas.

\begin{definition}
  Let $A = \sixtuple{\alphin}{\alphout}{\states}{\istate{0}}{\stransfun}{\soutfun}$
  be a finite-state transducer,
  and $w \in \alphin^*$ a word.
  Then~$A$ is \emph{predetermined~by~$w$} 
  if there exists $u \in \alphout^\nat$
  such that for every $w' \in \alphin^*$
  it holds that $\outfun{\istate{0}}{ww'} \prefixof u$.
\end{definition}

When a transducer is predetermined by $w$, 
then it transduces words starting with $w$ to ultimately periodic words (or finite words).
\begin{lemma}\label{lem:predetermined}
  Let $A = \sixtuple{\alphin}{\alphout}{\states}{\istate{0}}{\stransfun}{\soutfun}$
  be predetermined by $w \in \alphin^*$, and let $w' \in \str{\alphin}$.
  If the word $\outfun{\istate{0}}{ww'}$ is infinite, then it is ultimately periodic. 
\end{lemma}
\begin{proof}
  Let $u \in \alphout^\nat$ such that 
  \begin{align}
    \myal{w' \in \alphin^*}{\outfun{\istate{0}}{ww'} \prefixof u} \label{eq:prefix}
  \end{align}
  Let $w' \in \alphin^\nat$ such that $\outfun{\istate{0}}{ww'}$ is infinite.
  Note that from~\eqref{eq:prefix} it follows that $\outfun{\istate{0}}{ww'} = u$.
  Let $q = \transfun{\istate{0}}{w}$.
  Then $\outfun{\istate{0}}{ww'} = \outfun{\istate{0}}{w}\outfun{q}{w'}$.
  By the infinitary pigeonhole principle, there exists some state $q' \in Q$
  that is visited infinitely often when the automaton reads $w'$
  starting in state~$q$.
  Consequently there are non-empty words $w_0,w_1,\ldots \in \alphin^+$  such that
  $w' = w_0 w_1 w_2\cdots$, and
  for every $n\in\nat$ we have that $\transfun{q}{w_0w_1\cdots w_n} = q'$.
  As $\outfun{\istate{0}}{ww'}$ and hence $\outfun{q}{w'}$ is infinite, 
  there exists some $i > 0$ such that $\outfun{q'}{w_i} \ne \emptyword$.
  Define $v = \outfun{q}{w_0w_1\cdots w_{i-1} w_i w_i w_i w_i\cdots}$,
  then it follows that $v$ is infinite.
  Moreover $v$ is ultimately periodic as it is the transduct of an ultimately periodic sequence. 
  We obtain $\outfun{\istate{0}}{w}v = u$ from \eqref{eq:prefix},
  and thus $u = \outfun{\istate{0}}{ww'}$ is ultimately periodic.
\end{proof}

We are now ready to prove Theorem~\ref{thm:independent}.

\begin{proof}[Proof of Theorem~\ref{thm:independent}]
  Let $\mathcal{S} \subseteq \two^\nat$ be a countable set of words
  that contains no ultimately periodic words.
  Let $\mathcal{A}$ be the set of all finite-state transducers over the alphabet $\two$.
  Note that $\mathcal{A} \times \mathcal{S}$ is countable and
  let $(A_0,s_0), (A_1,s_1), \ldots$ be an enumeration of this set.
  For $i =0,1,2,\ldots$ we define
  words $w_i \in \two^+$ as follows.
    Let $v_i = w_0\cdots w_{i-1}$. We stipulate that $v_0 = \emptyword$.
    Let $A_i = \sixtuple{\two}{\two}{\states}{\istate{0}}{\stransfun}{\soutfun}$.
    If $A_i$ is predetermined by $v_i$, then the choice of $w_i$ is arbitrary; say $w_i = 0$.
    Otherwise, there exist words $x,y \in \alphin^+$ 
    such that neither $x' \prefixof y'$ nor $y' \prefixof x'$,
    where $x' = \outfun{\istate{0}}{v_ix}$ and $y' = \outfun{\istate{0}}{v_iy}$.
    Then there exists an index $j < \min \{|x'|,\, |y'|\}$ such that
    $x'(j) \ne y'(j)$.
    Define $w_i = y$ if $x'(j) = s_i(j)$, and $w_i = x$, otherwise. 
    This choice guarantees that
    \begin{align}
      \outfun{\istate{0}}{w_0\cdots w_{i}}  \not\prefixof s_i \label{eq:not:prefix}
    \end{align}
  Let   $w = w_0w_1w_2\cdots$.
  Assume that there exists $u \in \mathcal{S}$ with $w \fstred u$.
  Then there exist a finite-state transducer $A \in \mathcal{A}$ such that $w \red{A} u$.
  However, there is some $i \in \nat$ such that $(A_i,s_i) = (A,u)$.
  If $A_i$ is predetermined by $w_0w_1\cdots w_{i-1}$,
  then $u$ is ultimately periodic by Lemma~\ref{lem:predetermined}, and hence $u \not\in \mathcal{S}$.
  Otherwise property~\eqref{eq:not:prefix} contradicts 
  $w = w_0w_1w_2 \cdots \red{A} s_i = u$.
  Thus $w$ has the required properties.
\end{proof}

\section{Future Work}\label{sec:future}

\xtend{
We have shown that there are at least $\aleph_0$ many atoms in the hierarchy of transducer degrees.
They reside in a class of words over $\{0,1\}$
in which the distance of ones grows according to a polynomial.
In particular,
we have proven that, for every $k \ge 1$,
there exists a polynomial $p_k(n)$ of order $k$ 
such that the degree of the word $\seq{p_k(n)}$ is an atom (see Theorem~\ref{thm:atom}). 
This atom is the unique atom among,
and the infimum of, the degrees of polynomials of order $k$.

The degrees of $\seq{n}$ and $\seq{n^2}$
are the unique atoms among the polynomials of order $1$ and $2$, respectively.
Surprisingly, we find that the degree of $\seq{n^k}$ is \emph{not} an atom whenever $k \ge 3$.
The degree of $\seq{n^k}$ lies strictly above the degree of $\seq{p_k(n)}$.
}

Our results hint at an interesting structure of the transducer degrees of 
words $\seq{p(n)}$ for polynomials $p(n)$ of order $k \in \nat$.
Here, we have only scratched the surface of this structure.
Many questions remain open, for example:\smallskip
\begin{enumerate}\setlength{\itemsep}{.2ex}
  \item [\myquestion] \qlabel{q:spiralling:k} What is the structure of `polynomial spiralling' degrees (depending on $k \in \nat$)?
  Is the number of degrees finite for every $k \in \nat$?
  \item [\myquestion] \qlabel{q:nk} Are there interpolant degrees between the degrees of $\seq{n^k}$ and $\seq{p_k(n)}$?
  \item [\myquestion] \qlabel{q:atoms} Are there continuum many atoms?
  \item [\myquestion] \qlabel{q:morse} Is the degree of the Thue--Morse sequence an atom?
\end{enumerate}

%
%
%
%

\bibliography{main}

\xtend{
\newpage
\appendix

\section*{Appendix}

\section{Weighted Products}\label{sec:composition}

We define concatenation and unfolding of tuples of weights.
\begin{definition}\normalfont
  Let $\vec{\alpha} = \tup{\vec{\alpha_0},\ldots,\vec{\alpha_{m-1}}}$,
  $\vec{\beta} = \tup{\vec{\beta_0},\ldots,\vec{\beta_{m-1}}}$ be a tuple of weights.
  We define \emph{concatenation}:
  \begin{align*}
    \vec{\alpha} \concat \vec{\beta} = \tup{\vec{\alpha_0},\ldots,\vec{\alpha_{m-1}},\vec{\beta_0},\ldots,\vec{\beta_{m-1}}}\;.
  \end{align*}
  We define \emph{unfolding} by induction on $n \in \nat$ with $n > 0$: 
  \begin{align*}
    \vec{\alpha}^1 = \vec{\alpha} && \vec{\alpha}^{n+1} = \vec{\alpha} \concat \vec{\alpha}^n
  \end{align*}
\end{definition}

Unfolding a tuple of weights
does not change its semantics.
\begin{lemma}
  Let $f : \nat \to \rat$, $\vec{\alpha}$ a tuple of weights and $n \ge 1$.
  Then $\wprod{\vec{\alpha}}{f} = \wprod{\vec{\alpha}^n}{f}$.
\end{lemma}
\begin{proof}
  Follows immediately from the cyclic fashion in which the weights in the weighted product are applied.
\end{proof}

We will now define the product $\wprod{\vec{\alpha}}{\vec{\beta}}$ of tuples of weights 
such that we have
$\wprod{\vec{\alpha}}{(\wprod{\vec{\beta}}{f})} = \wprod{(\wprod{\vec{\alpha}}{\vec{\beta}})}{f}$
for every $f : \nat \to \rat$.
We need one auxiliary definition.
\begin{definition}\label{def:prod:weight:weights}\normalfont
  For a weight $\vec{\gamma} = \tup{x_{0},\ldots,x_{\ell-1},y}$
  and a tuple of weights $\vec{\alpha} = \tup{\vec{\alpha_0},\ldots,\vec{\alpha_{\ell-1}}}$ 
  with $\vec{\alpha_i} = \tup{a_{i,0},\ldots,a_{i,m_i},b_i}$,
  we define the weight $\vec{\gamma} \cdot \vec{\alpha}$ by
  \begin{align*}
    \vec{\gamma} \cdot \vec{\alpha} &= \tup{x_0a_{0,0},\; \ldots,\; x_0a_{0,m_0}, \\
      &\hspace{3.8ex}x_1a_{1,0},\; \ldots,\; x_1a_{1,m_1}, \\[-1ex]
      &\hspace{12ex}\vdots\\[-1ex]
      &\hspace{3.8ex}x_\ell a_{\ell,0},\; \ldots,\; x_\ell a_{\ell,m_\ell},\\
      &\hspace{3.8ex} x_0 b_0 + x_1 b_1 + \cdots + x_\ell b_\ell + y
      }\;.
  \end{align*}
\end{definition}

We now define the product of tuples of weights.
\begin{definition}\label{def:wprod:compose}\normalfont
  For tuples of weights $\vec{\alpha}$ and $\vec{\beta}$ 
  with the property $\sumlength{\vec{\alpha}} = \length{\vec{\beta}}$
  we define $\wprod{\vec{\alpha}}{\vec{\beta}}$ by induction on the tuple length:
  \begin{align*}
    \wprod{\vec{\alpha}}{\vec{\beta}} =\ & 
      \tup{\vec{\alpha_0}\cdot \tup{\vec{\beta_0},\ldots,\vec{\beta_{|\alpha_0|-2}}}} \concat \\
      &\quad \big(\wprod{\tup{\vec{\alpha_1},\ldots,\vec{\alpha_{k-1}}}}{\tup{\vec{\beta_{|\alpha_0|-1}},\ldots,\vec{\beta_{\ell-1}}}}\big)
  \end{align*}
  where $\vec{\alpha} = \tup{\vec{\alpha_0},\ldots,\vec{\alpha_{k-1}}}$.
  and $\vec{\beta} = \tup{\vec{\beta_0},\ldots,\vec{\beta_{\ell-1}}}$.
  Here we stipulate that $\wprod{\tup{}}{\tup{}} = \tup{}$.
  
  For $\vec{\alpha}$ and $\vec{\beta}$ 
  such that $\sumlength{\vec{\alpha}} \ne \length{\vec{\beta}}$
  we define $\wprod{\vec{\alpha}}{\vec{\beta}}$
  as follows:
  \begin{align*}
    \wprod{\vec{\alpha}}{\vec{\beta}} =     
    \wprod{\big(\vec{\alpha}^{\frac{c}{\sumlength{\vec{\alpha}}}}\big)}{\big(\vec{\beta}^{\frac{c}{\length{\vec{\beta}}}}\big)}
  \end{align*}
  where $c \in \nat$ is the least common multiple of $\sumlength{\vec{\alpha}}$ and $\length{\vec{\beta}}$.
\end{definition}

\begin{example}
  Let $\vec{\alpha} = \tup{\vec{\alpha_1},\vec{\alpha_2}}$ and $\vec{\beta} = \tup{\vec{\beta_1},\vec{\beta_2}}$
  \begin{align*}
    \vec{\alpha_1} &= \tup{2,1,3} & \vec{\alpha_2} &= \tup{1,1} \\
    \vec{\beta_1} &= \tup{1,2,3,4} & \vec{\beta_2} &= \tup{0,1,1}
  \end{align*}
  Note that $\vec{\beta}$ is the tuple of weights used in Example~\ref{ex:wprod}. 
  We compute $\wprod{\vec{\alpha}}{\vec{\beta}}$.
  We have $\sumlength{\vec{\alpha}} = 3$ and $\length{\vec{\beta}} = 2$.
  Thus, we have to unfold $\vec{\alpha}$ twice and $\vec{\beta}$ trice:
  $\vec{\alpha}^2 = \tup{\vec{\alpha_1},\vec{\alpha_2},\vec{\alpha_1},\vec{\alpha_2}}$
  and
  $\vec{\beta}^3 = \tup{\vec{\beta_1},\vec{\beta_2},\vec{\beta_1},\vec{\beta_2},\vec{\beta_1},\vec{\beta_2}}$.
  Then
  \begin{align*}
    &\wprod{\vec{\alpha}}{\vec{\beta}} 
    = \wprod{\vec{\alpha}^2}{\vec{\beta}^3} \\
    &= \wprod{\tup{\vec{\alpha_1},\vec{\alpha_2},\vec{\alpha_1},\vec{\alpha_2}}}{\tup{\vec{\beta_1},\vec{\beta_2},\vec{\beta_1},\vec{\beta_2},\vec{\beta_1},\vec{\beta_2}}} \\
    &= \tup{\vec{\alpha_1} \cdot \tup{\vec{\beta_1},\vec{\beta_2}}} \concat \wprod{\tup{\vec{\alpha_2},\vec{\alpha_1},\vec{\alpha_2}}}{\tup{\vec{\beta_1},\vec{\beta_2},\vec{\beta_1},\vec{\beta_2}}} \\[-.5ex]
    &\hspace{5ex}\vdots\\[-.5ex]
    &= \tup{\vec{\alpha_1} \cdot \tup{\vec{\beta_1},\vec{\beta_2}}} \concat \tup{\vec{\alpha_2} \cdot \tup{\vec{\beta_1}}} 
       \concat \tup{\vec{\alpha_1} \cdot \tup{\vec{\beta_2},\vec{\beta_1}}} \concat \tup{\vec{\alpha_2} \cdot \tup{\vec{\beta_2}}} \\
    &= \tup{\tup{2,4,6,0,1,12},\; \tup{1,2,3,5},\; \tup{0,2,1,2,3,9},\; \tup{0,1,2} }
  \end{align*}
\end{example}

\begin{proof}[Proof of Lemma~\ref{lem:wprod:double:non:constant}]
  Follows directly from the definition of $\wprod{\vec{\alpha}}{\vec{\beta}}$.
  Every weight in $\wprod{\vec{\alpha}}{\vec{\beta}}$ is a concatenation of 
  scalar multiplications of weights in $\vec{\beta}$.
  \qed
\end{proof}

\section{Additional Intuition for the Proof of Theorem~\ref{thm:atom}}\label{sec:intuitive:atom}

Let $k \ge 1$ and $a_0,\ldots,a_{k-1} > 0$. Define 
\begin{align*}
  p_k(n) = \sum_{i = 0}^{k - 1} a_i (k n + i)^k \;.
\end{align*}
Theorem~\ref{thm:atom} states that
for every polynomial $q(n) \in \qpoly_k$,
we have $\seq{q(n)} \fstred \seq{p_k(n)}$.
Hence the degree of $\seq{p(n)}$ is an atom as a consequence of Proposition~\ref{prop:poly-transducts}.

We give some more intuition for the proof of Theorem~\ref{thm:atom},
involving a more explicit appeal to the continuity argument.
For functions $f_0,\ldots,f_{k-1} : \nat \to \rat$,
we define the function $\zip{f_0,\ldots,f_{k-1}} : \nat \to \rat$ 
a.k.a.\ \emph{perfect shuffle}~\cite{allo:shal:1999,ramp:shal:wang:2005}, 
by induction on $n$ as follows
\begin{align*}
  \zip{f_0,\ldots,f_{k-1}}(0) &= f_0(0) \\
  \zip{f_0,\ldots,f_{k-1}}(n+1) &= \zip{f_1,\ldots,f_{k-1},\shift{1}{f_0}}(n)
\end{align*}
where $\shift{i}{f}$ is the \emph{$i$-th shift} of $f$, defined by $m \mapsto f(m+i)$.

We have $\seq{n^k} \fstred \seq{p_k}$ by the following transduction:
\begin{align*}
  \seq{n^k} &= \seq{\zip{(k n + 0)^k, (k n + 1)^k, \ldots, (k n + k-1)^k}} \\
            &\fstred \seq{a_0 (k n + 0)^k + \cdots + a_{k-1}(k n + k-1)^k} \\
            &= \seq{p_k(n)}
\end{align*}
Thinking of $n^k$ as an infinite word of natural numbers,
then $(k n + 0)^k$, $(k n + 1)^k$, \ldots, $(k n + k-1)^k$
are subsequences of $n^k$.
Namely those subsequences picking every $k$-th element starting from element at index $0$,$1$,\ldots,$k-1$,
respectively.
Note that the transduction in the second line corresponds to the weighted product $\tup{a_0,a_1,\ldots,a_{k-1},0}$,
and thus can be realised by a finite state transducer (Lemma~\ref{lem:wprod:FST:q}).

However, this transduction works only for $\seq{n^k}$.
It remains to be argued that there exists such a transduction from $\seq{q(n)}$ to $\seq{p_k(n)}$ 
for every polynomial $q(n) \in \qpoly_k$,

Let us write $\sim_\varepsilon$ for the relation that relates polynomials of the same order
whose coefficients differ by at most $\varepsilon > 0$.
By Lemma~\ref{lem:epsilon} we can get arbitrarily close to~$\seq{n^k}$.
For every $\varepsilon > 0$, there exists $h(n) \in \qpoly_k$ such that  
\begin{align*}
  \seq{q(n)} \fstred \seq{h(n)} &&\text{and}&& h(n) \sim_\varepsilon n^k
\end{align*}
Moreover, for $i \in \nat_{<k}$, we have 
\begin{align}
  h(k n + i) \sim_{\varepsilon'} (k n + i)^k \label{sub:pn:sim}
\end{align}
where $\varepsilon'$ depends on $\varepsilon$ (and $i$).
If $\varepsilon$ tends to $0$, so will $\varepsilon'$.

The crucial observation is that 
$(k n + 0)^k$, $(k n + 1)^k$, \ldots, $(k n + k-1)^k$
form a basis of the vector space of polynomials of order $k$ with addition and scalar multiplication.\footnote{%
  The cautious reader will have observed that the basis consists only of $k$ vectors
  while the vector space has dimension $k+1$.
  We tacitly ignore the constant terms of the polynomials
  since, in the transducer degrees, we have $\seq{f} \equiv \seq{f + c}$ for every $c\in\rat$.
}
The property of `being a basis' is continuous.
Hence, for small enough $\varepsilon$, and using approximation~\eqref{sub:pn:sim},
we conclude that 
$h(k n + 0)$, $h(k n + 1)$, \ldots, $h(k n + k-1)$ 
form a basis as well.
Thus, there exist $a_0',a_1',\ldots,a_{k-1}' \in \rat$ such that
\begin{align*}
  p_k(n) = a_0' h(k n + 0) + \cdots + a_{k-1}'h(k n + k-1)
\end{align*}
We have that
\begin{align*}
  \seq{h(n)} &= \seq{\zip{h(k n + 0), h(k n + 1), \ldots, h(k n + k-1)}}\\
            &\fstred \seq{a_0' h(k n + 0) + \cdots + a_{k-1}'h(k n + k-1)} \\
            &= \seq{p_k(n)}
\end{align*}
However, for this transduction to work, we need to ensure that  
$a_0',a_1',\ldots,a_{k-1}' \ge 0$.
Recall that $a_0,a_1,\ldots,a_{k-1} > 0$.
Again, by continuity,  
$a_i'$ approaches $a_i$ as $\varepsilon$ approaches $0$.
Hence, for small enough $\varepsilon$,
we have $a_i' \ge 0$ for every $i \in \nat_{<k}$.
Thus we have $\seq{q(n)} \fstred \seq{h(n)} \fstred \seq{p(n)}$.
Hence $p(n)$ is an atom.

}

\end{document}